\documentclass[conference,letterpaper]{IEEEtran}

\usepackage{amssymb,amsfonts,bm}
\usepackage{graphicx}
\usepackage{color}
\usepackage{upgreek,enumitem}
\usepackage{algorithm}
\usepackage{algpseudocode}

\usepackage[utf8]{inputenc} 
\usepackage[T1]{fontenc}
\usepackage{url}
\usepackage{ifthen}
\usepackage{cite}
\usepackage[cmex10]{amsmath} 

\newcommand{\be}{\begin{equation}}
\newcommand{\ee}{\end{equation}}
\newcommand{\ba}{\begin{align}}
\newcommand{\ea}{\end{align}}

\newcommand{\summ}[3]{\sum_{#1=#2}^{#3}}

\newtheorem{theorem}{Theorem}
\newtheorem{definition}{Definition}[section]
\newtheorem{corollary}{Corollary}
\newtheorem{remark}{Remark}[section]
\newtheorem{lemma}{Lemma}[section]
\newtheorem{claim}{Claim}[section]

\newenvironment{proof}[1][Proof]{\noindent\textit{#1:} }{\hfill$\blacksquare$\par}

\begin{document}
\title{Secure (Multiple) Key-Cast over Networks: \\Multiple Eavesdropping Nodes}


\author{%
  \IEEEauthorblockN{Reza Sayyari}
  \IEEEauthorblockA{
                    University at Buffalo\\
                    {\tt rezasayy@buffalo.edu}
                    }
  \and
  \IEEEauthorblockN{Michael Langberg}
  \IEEEauthorblockA{
                    University at Buffalo\\
                    {\tt mikel@buffalo.edu}}
}

\maketitle


\begin{abstract}
We study the secure multiple key-cast problem over noiseless networks under node-based eavesdroppers, where one or more source nodes participate in the generation of distinct secret keys to be shared among designated terminal subsets, while an eavesdropper observing up to $\ell$ nodes, including possibly source nodes, obtains no information about the keys. 
For the single-source setting, we first consider networks in which every node is $d$-vertex connected from the source. We show that a secure key rate of $d-\ell$ is achievable for all such networks. We further show that this rate is optimal by exhibiting $d$-vertex-connected networks whose secure key-cast capacity is at most $d-\ell$.
We next study networks in which only the terminal nodes are 
\(d\)-vertex connected from the source, while other network nodes may not satisfy this connectivity condition and may be {\em partially-connected}.
We show that secure multiple key-cast remains achievable in the presence of such partially-connected nodes, and derive coding schemes whose rate depends on the minimum network vertex-connectivity from the source and certain additional network properties.
Finally, we generalize these results, for both $d$-vertex-connected networks and networks containing partially-connected nodes, to the multi-source setting; showing that secure multiple key-cast remains achievable even when the eavesdropper may observe all but one of the source nodes.
\end{abstract}

\section{Introduction}
Shared randomness is a fundamental resource in secure communication and distributed computation. A wide range of operations, including encryption, authentication, randomized coding, distributed inference, federated learning, privacy-preserving computation, identification, and local differential privacy, benefit when terminals share uniform secret keys across designated subsets of nodes, e.g., \cite{langberg2022network, 10904071, langberg2024characterizing, bhattacharya2019shared, 10482871} and references therein. In large-scale systems such as federated learning and the Internet of Things (IoT), shared randomness enables coordinated operation with reduced communication overhead \cite{byrd2020differentially, bonawitz2017practical}. As networks grow and adversarial exposure increases, disseminating secret keys securely becomes a central challenge. Nodes may be geographically dispersed, adversaries may observe parts of the network, and communication resources may be limited, making the secure and efficient distribution of shared keys essential for modern network design.

The work at hand studies the task of secure (secret) key dissemination in the context of noiseless networks, i.e., in the context of network coding \cite{ahlswede2000network, li2003linear}. 
We study a number of different communication settings including single or multiple source nodes that generate randomness, and single or multiple terminal sets that require to share a secret key. 
In these settings, shared randomness is generated at network source nodes and is distributed through designed coding operations across the network; terminal nodes belong to disjoint subsets, and those that lie in the same subset receive a common key; the key is kept secret from an adversarial eavesdropper that controls at most $\ell$ network nodes. 
When all terminals belong to the same set (i.e., a common key is distributed to all terminals), we call our problem the (single or multiple-source) {\em key-cast} problem,
and when there are several terminal sets,  (single or multiple-source) {\em multiple key-cast}.

Our work makes several contributions toward understanding and designing secure (multiple) key-cast protocols under node-based eavesdropping. In Section~\ref{SysM}, we define our multiple key-cast model and introduce the terminology used throughout the paper. Section \ref{LemmSec} establishes the preliminary lemmas required for our main results. Building on these foundations, in Section~\ref{FCsec}, we study networks in which all nodes satisfy a certain connectivity condition, and show that optimal-rate secure multiple key-cast is achievable using a scheme based on the secure regenerating codes proposed in \cite{shah2011information}. Then in Section~\ref{PCsec}, we relax this connectivity requirement  to accommodate intermediate nodes that may only be partially-connected and present coding solutions whose rate depends on the network vertex-connectivity from the source(s) and certain additional network properties.
The task of secure (multiple) key-cast has strong connections to the study of secure network coding, the study of regenerating codes, and the study of secret sharing.
In what follows, we give an overview of these connections including previous results in the context of secure key-dissemination.

\subsection{Relation to prior works}
The task of secure key-dissemination, also called secret key-agreement, has been mostly studied over noisy memoryless channel structures, where secrecy is ensured through channel noise, e.g.,  \cite{256484, csiszar2003broadcast, tyagi2015converses}. In the noisy setting, the goal is to reconstruct a shared secret key from correlated channel observations at the legitimate users while ensuring negligible information leakage to the eavesdropper.
In these methods, the amount of secret key that can be generated is tied to the statistical dependence induced by the channel among the users’ observations \cite{256484}. 
This work studies the problem of secure (multiple) key-cast, where shared randomness is distributed securely among authorized terminals over large networks in the context of network coding.

\paragraph{Secure network coding} 
The problem of secure network coding, e.g., \cite{cai2002secure,cai2010secure,feldman2004capacity}, is closely related to the (multiple) key-cast problem, where the key(s) themselves can be treated as a confidential message(s).
In the {\em single-source} setting, secure key-cast, i.e., distributing a secret key to {\em one} terminal set, is shown in \cite{10904071} to be equivalent to secure multicast network coding. 
On the other hand, with multiple sources or terminal sets, multiple key-cast differs from secure network coding since terminals no longer need to reconstruct source messages and the keys may be mixtures of source information. Moreover, in the multi-source key-cast setting, source nodes may also be observed by the eavesdropper, so even access to all but one source reveals no information about the keys. In contrast, secure multi-source network coding assumes each source knows its message and excludes eavesdroppers observing source nodes. An illustrative example of a network for which secure multicast is impossible, yet key-cast remains achievable is given in Figure \ref{fig:1}.
We next provide an overview of existing results in the context of secure network coding for edge-based and node-based eavesdroppers.
 
\begin{figure}
    \centering
        \includegraphics[width=0.5\linewidth]{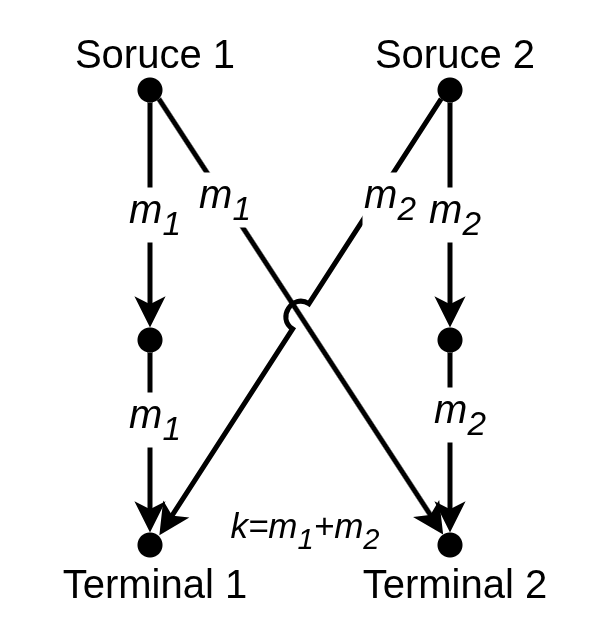}
    \caption{Network example in which secure multicast is infeasible but secure key-cast is possible. Two source nodes generate independent, uniformly distributed symbols $m_1$ and $m_2$. Although $m_1$ and $m_2$ cannot be delivered securely to the terminals individually under a node-based eavesdropper model, both terminals can recover the shared key $K = m_1 + m_2$ while any eavesdropper observing a single node gains no information about the key.}
        \label{fig:1}
\end{figure}

Secure network coding with information-theoretic guarantees was first introduced in \cite{cai2002secure,cai2010secure}, where it was shown that in a single-source multicast network, perfect secrecy against an eavesdropper tapping up to {\bf $\ell$ edges} can be achieved at a rate equal to the source-terminal min-cut minus $\ell$.
Later, \cite{feldman2004capacity} demonstrated that random linear network coding achieves this secrecy capacity with high probability over large fields.
{Once several nodes can generate randomness to be used in the communication process,
and/or eavesdroppers have access to  edge sets  with varying capacities (e.g., the setting of node eavesdroppers), the capacity is not fully characterized; its determination, in certain cases, is known to be NP-hard \cite{cui2012secure} or as hard as solving the general multiple-unicast problem \cite{8327896}, which is a long standing open problem, e.g., \cite{dougherty2005insufficiency, kamath2016two}.

For the more general and challenging setting of {\bf node-based} eavesdroppers, much less is known.
\cite{lima2007random} studies security in random linear multicast network coding, and shows that intermediate nodes cannot decode the transmitted message if  its in-degree is smaller than the number of source symbols.
\cite{che2013routing} studies secure {\em unicast} from a source to a single terminal under node-based eavesdropping (and jamming) and characterizes the optimal secure rate among routing-only schemes. It leaves open the study of more general network coding schemes.
In \cite{zhang2013lightweight}, the authors study secure network coding from a different perspective, focusing on the efficiency–security trade-off rather than achievable rate, and proposes a lightweight encryption scheme under an adversarial model that can observe combinations of links and nodes. However, the approach relies on pre-shared secret keys between different network nodes (which is the focus of this work).
\cite{wang2016optimal} studies information-theoretically secure unicast streaming without encryption or pre-shared keys, and designs both deterministic and random linear network coding schemes that jointly account for security, rate, randomness overhead, and bandwidth cost under a node-based adversary model. 
However, the threat model in \cite{wang2016optimal}, similar to that in \cite{lima2007random}, considers only a single eavesdropping node or non-colluding compromised nodes.
The model and results of the work at hand differ significantly from those above in the sense that we consider the task of key distribution and do not require source reconstruction, we consider single and multiple source nodes that can generate network randomness, we consider multiple terminal sets, each decoding its own key, and we consider multiple colluding eavesdropping nodes.

\paragraph{Previous results on (multiple) key-cast using regenerating codes and secret sharing} Regenerating codes provide a structured framework for distributed storage systems, in which a file is divided into multiple fragments and stored across different nodes, allowing a failed node to recover its data by connecting to a subset of surviving nodes. Although originally developed for storage applications, many of their underlying principles are closely related to secure aspects of network coding.
Regenerating codes were introduced in \cite{dimakis2010network} to enable efficient repair in distributed storage by reducing the bandwidth required to reconstruct lost data. By representing node contents as collections of finite-field symbols, these codes allow a failed node to recover its data by downloading linear combinations of small fragments from surviving nodes rather than retrieving the entire file. A secure variant was later proposed in \cite{shah2011information}, where additional randomness is embedded into stored and transmitted symbols so that an eavesdropper observing a bounded number of nodes, as well as the repair data, gains no information about the stored file.

In secret sharing, a secret is divided into multiple shares and distributed among different nodes such that authorized subsets can reconstruct the secret, while unauthorized subsets obtain no information about it. Building on secure regenerating codes, \cite{shah2015distributed} constructed a secret sharing scheme for networks in which the source cannot directly reach all participants. Their scheme disseminates shares in a fully distributed manner using only local neighborhood information, provides information-theoretic security against eavesdroppers observing fewer nodes than the reconstruction threshold, and enables recovery of the secret by all authorized terminals while reducing both communication overhead and randomness requirements.

From a (multiple) key-cast perspective, secure regenerating codes and secret sharing schemes can be viewed as mechanisms for distributing shared randomness across network nodes in the presence of node-based adversaries. Rather than reconstructing a specific file, the objective is to ensure that authorized nodes collectively obtain sufficient information to agree on a common secret, while an adversary with limited node access observes only statistically independent combinations.
Building on this perspective, \cite{langberg2022network, 10904071} utilizes the ideas of the secret sharing scheme proposed in \cite{shah2015distributed} to construct a multiple key-cast protocol for networks with a single source node.
The scheme suggested in \cite{langberg2022network, 10904071} allows to disseminate keys of unit rate as long as every terminal has two vertex-disjoint (unit capacity) paths from the source, and every non-terminal node has two edge-disjoint (unit capacity) paths from the source. The disseminated keys are kept secret from an eavesdropper capable of observing all information of any single, i.e., $\ell=1$, non-terminal node (except the source node).
The work at hand leverages the connections established between secret sharing, regenerating codes, and key-cast to generalize the results of \cite{langberg2022network, 10904071} to eavesdroppers that control multiple network nodes, i.e., $\ell$ eavesdropping nodes for general values of $\ell$.
Additional prior work on the task of key-dissemination in the context of network coding includes \cite{langberg2024characterizing} that characterizes network conditions for positive-rate in the setting of secure multicast and key-cast in the presence of an eavesdropper that may observe any single non-terminal node.

\section{System Model}\label{SysM}
\paragraph{Notation}
Matrices are denoted by uppercase boldface letters (e.g., $\boldsymbol{A}$), vectors by lowercase boldface letters (e.g., $\boldsymbol{a}$), and scalars by italic letters (e.g., $a$). The finite field of size $q$ is denoted by $\mathbb{F}_q$. A set $\mathcal{X}$ is a collection of distinct elements, and its cardinality is denoted by $|\mathcal{X}|$. The operator $(\cdot)^T$ denotes transpose. For vectors $\mathbf{a}, \mathbf{b} \in \mathbb{F}_q^n$, their inner product is written as $\mathbf{a}^T \mathbf{b} = \sum_{i=1}^n a_i b_i$. For a vector $\boldsymbol{a} \in \mathbb{F}_q^n$, we denote by $[\boldsymbol{a}]_{i:j}$ the subvector containing components $i$ through $j$ of $\boldsymbol{a}$. For a matrix $\boldsymbol{A}\in \mathbb{F}_q^{n \times m}$, the entry in the $i$-th row and $j$-th column is denoted by $[\boldsymbol{A}]_{i,j}$. For any integer $\ell$, we define $[\ell] = \{1, 2, \ldots, \ell\}$. The vector $\boldsymbol{e}_i$ denotes the $i$th standard basis vector, with a single nonzero entry equal to $"1"$ in position $i$ and zeros elsewhere. $\boldsymbol{0}_{m \times 1}$ denotes the all-zero vector of length $m$. The neighborhood of a node $n$, denoted by $\mathcal{N}(n)$, is the set of nodes that are direct outgoing neighbors of $n$ (i.e., nodes connected by an edge from $n$). The set of parent nodes of a node $j$, denoted by $\mathcal{P}(j)$, is the set of nodes with direct edges leading into $j$.

The network instance $\mathcal{I}=\left(\mathcal{V}, \mathcal{E}, \mathcal{B}, \mathcal{S}, \mathcal{T}\right)$ is represented by a directed acyclic graph (DAG)
$\mathcal{G} = (\mathcal{V}, \mathcal{E})$, 
where $\mathcal{V} = \{v_1,\dots,v_{|\mathcal{V}|}\}$ is the set of nodes,
$\mathcal{E}$ is the set of directed edges, 
$\mathcal{S} \subseteq \mathcal{V}$ is the set of source nodes, where each source $s \in {\mathcal S}$ holds an unlimited amount of independent random symbols, denoted by $X_s$, drawn from $\mathbb{F}_q$,
and $\mathcal{T} \subseteq \mathcal{V}$ is the set of terminal nodes.
The set $\mathcal{T}$ is partitioned as
$\mathcal{T}=\{\mathcal{T}_i\}_{i=1}^m$ into $m$ pairwise disjoint subsets $\mathcal{T}_1,\mathcal{T}_2,\dots,\mathcal{T}_m$,
where for each $i=1,\dots, m$, terminals in $\mathcal{T}_i$ decode the same key (distinct from the keys decodes in other subsets $\mathcal{T}_j$, $j \ne i$).
The security requirement is represented by $\mathcal{B}=\{\beta_{1},\dots,\beta_{|\mathcal{B}|}\}$, where each $\beta \in \mathcal{B}$ is a subset of nodes that can be observed by the eavesdropper.
We assume that source nodes have no incoming edges and that terminal nodes have no outgoing edges. All operations in the network are assumed to be performed over a finite field $\mathbb{F}_q$ of size $q$. Every edge $e \in \mathcal{E}$ has unit capacity, i.e., each edge can carry a single symbol from $\mathbb{F}_q$. To model higher capacity edges, multiple (unit-capacity) edges may exist between two given nodes.

\paragraph{Key-Codes}
For a blocklength $n$, a \emph{key-code} $\mathcal{C} = (\mathcal{F}, \mathcal{G})$ consists of a collection of local encoding functions $\mathcal{F} = \{ f_e : e \in \mathcal{E} \}$ and decoding functions $\mathcal{G} = \{ g_{t} : t \in \mathcal{T}\}$. For each edge $e=(u,v) \in \mathcal{E}$, the transmitted message $X_e^n \in [q^n]$ is obtained by applying the encoding function $f_e$ to the information available at node $u$. More precisely, for a generic node $u \in \mathcal{V}\setminus \mathcal{S}$, the set of available inputs is given by
\(
X_{u} = \big( (X_{e'}^n : e' = (v,u) \in \mathcal{E}) \big),
\)
which captures all symbols received by $u$. 
For a source $s$, as stated above, $X_s$ includes the randomness available at $s$. To ensure well-defined encoding, communication proceeds according to a topological ordering.

\paragraph{Secure (Multiple) Key-Cast Instance}
An instance $\mathcal{I}$ allows secure multiple key-cast if there exists a key-code that enables a collection of keys $\{K_1, K_2, \dots, K_m\}$ to be securely transmitted, where each key $K_i$ is intended for terminal set $\mathcal{T}_i \subseteq \mathcal{T}$, while for each $i \in \{1,\dots,m\}$, an eavesdropper observing nodes in any $\beta \in \mathcal{B}$ obtains no information about $K_i$. The special case $m=1$, in which there is a single terminal set, is referred to as a secure key-cast instance. More specifically,
\begin{itemize}
    \item Each terminal in $\mathcal{T}_i$ can recover its corresponding key $K_i$:
    \[
        \forall i \in \{1,\dots,m\},\ \forall t \in \mathcal{T}_i, \quad H\!\left(K_i \mid X_t\right) = 0 .
    \]
    
    \item For each $i \in \{1,\dots,m\}$ and every $\beta \in \mathcal{B}$ such that $\beta \cap \mathcal{T}_i =\phi$, the eavesdropper obtains no information about $K_i$, i.e.,
    \(
        I\!\left(K_i ; \{X_v : v \in \beta\}\right) = 0 .
    \)
\end{itemize}

\begin{definition}[Key-Rate]
    For $i \in \{1,\dots,m\}$, let $K_i$ denote the key transmitted from the source nodes to a set of terminals $\mathcal{T}_i \in \mathcal{T}$ over a network $\mathcal{G} = (\mathcal{V}, \mathcal{E})$, and let $n$ denote the code blocklength. The network key-rate $R$ is defined as
    $$
    R = \min_{i \in \{1,\dots,m\}}{\frac{H(K_i)}{n}}.
    $$
\end{definition}

\begin{definition}[$d$-Vertex Connected]
    A node $v \in \mathcal{V}$ is said to be $d$-vertex connected from a source $s \in \mathcal{S}$ if for any integer $0 \leq c \leq d$,
    there exist at least $c$ vertex-disjoint paths from $s$ to $v$ and at least $d-c$ edges connecting $s$ to $v$.
    For any $v \notin \mathcal{N}(s)$, being $d$-vertex connected implies that the removal of any set of fewer than $d$ intermediate nodes (excluding $s$ and $v$) does not disconnect $v$ from $s$.
\end{definition}

Throughout, we associate with each node $v_i$ in $\mathcal{V}$ a distinct element $\alpha_i$ of $\mathbb{F}_q$ and a corresponding Vandermonde vector.
\begin{definition}[Vandermonde vector]
    The Vandermonde vector of size $d$ assigned to node $v_i$ is defined as    
    $$
    \boldsymbol{v}_i = \big(1, \alpha_i, \alpha_i^2, \ldots, \alpha_i^{d-1}\big) \in \mathbb{F}_q^d,
    $$    
    where $\alpha_i$ is a distinct element associated with node $v_i$, chosen from the underlying field $\mathbb{F}_q$.
\end{definition}

\section{Preliminary Lemmas}\label{LemmSec}
We establish a set of theorems showing the existence of key rates for Secure Multiple Key-Cast instances that satisfy certain connectivity requirements. 
First, we show that if all nodes in the network are $d$-vertex connected from the source, then the optimal key rate is achievable. 
Next, we present results for which the connectivity requirements are relaxed that demonstrate the existence of Multiple Key-Cast schemes for networks where not all nodes satisfy the $d$-vertex connectivity property. 
We start by presenting a umber of preliminary lemmas we use in our theorem proofs.

\begin{lemma}\label{Lemmfl}
    Let $\boldsymbol{M}\in\mathbb{F}_q^{d\times d}$ be a random symmetric matrix whose upper–triangular entries are independent and uniformly distributed over~$\mathbb{F}_q$. Let $\boldsymbol{U}=[\boldsymbol{V}_1,\boldsymbol{V}_2]$, where $\boldsymbol{V}_1\in\mathbb{F}_q^{d\times \ell}$, $\boldsymbol{V}_2\in\mathbb{F}_q^{d\times (d-\ell)}$, and $\boldsymbol{U}$ is invertible. Suppose that $\boldsymbol{M}\boldsymbol{V}_1=\boldsymbol{0}$. Then
    \[
    \boldsymbol{K}=\boldsymbol{U}^T\boldsymbol{M}\boldsymbol{U}
    \]
    is symmetric and has the following block form
    \[
    \boldsymbol{K}=
    \begin{bmatrix}
        \boldsymbol{0}_{\ell\times \ell} & \boldsymbol{0}_{\ell\times (d-\ell)}\\[0.3em]
        \boldsymbol{0}_{(d-\ell)\times \ell} & \boldsymbol{G}
    \end{bmatrix},
    \]
    where $\boldsymbol{G}\in\mathbb{F}_q^{(d-\ell)\times (d-\ell)}$ is symmetric and its upper–triangular entries are independent and uniformly distributed over~$\mathbb{F}_q$.
\end{lemma}

\begin{proof}
    Let $\mathcal{S}_d$ denote the set of all symmetric $d\times d$ matrices over $\mathbb{F}_q$. Its cardinality is $|\mathcal{S}_d| = q^{d(d+1)/2}$. Since $\boldsymbol{M}$ is drawn uniformly from $\mathcal{S}_d$, its probability mass function satisfies
    \[
    \text{Prob}(\boldsymbol{M}=\boldsymbol{M}_0)=\frac{1}{|\mathcal{S}_d|}, \qquad \forall\, \boldsymbol{M}_0\in\mathcal{S}_d.
    \]
    
    Because $\boldsymbol{M}$ is symmetric,
    \[
    \boldsymbol{K}^T = (\boldsymbol{U}^T \boldsymbol{M}\boldsymbol{U})^T
    = \boldsymbol{U}^T \boldsymbol{M}\boldsymbol{U}
    = \boldsymbol{K},
    \]
    and therefore $\boldsymbol{K}\in\mathcal{S}_d$ is also symmetric.
    
    Consider the mapping $f: \mathcal{S}_d \to \mathcal{S}_d$ defined by $f(\boldsymbol{A}) = \boldsymbol{U}^T \boldsymbol{A} \boldsymbol{U}$.
    Since $\boldsymbol{U}$ is invertible, this mapping is a bijection and the inverse mapping is given by 
    $$f^{-1}(\boldsymbol{B}) = (\boldsymbol{U}^T)^{-1} \boldsymbol{B} \boldsymbol{U}^{-1} = (\boldsymbol{U}^{-1})^T \boldsymbol{B} \boldsymbol{U}^{-1},$$
    where the second equality follows from the fact that $\boldsymbol{U}$ is invertible. Because $f$ is a bijection and symmetry is preserved, the random matrix $\boldsymbol{K} = \boldsymbol{U}^T \boldsymbol{M} \boldsymbol{U}$ is also uniformly distributed over $\mathcal{S}_d$. Thus, for any $\boldsymbol{K}_0 \in \mathcal{S}_d$,
    \[
    \text{Prob}(\boldsymbol{K} = \boldsymbol{K}_0) = \text{Prob}(\boldsymbol{M} = f^{-1}(\boldsymbol{K}_0)) = \frac{1}{|\mathcal{S}_d|}.
    \]

    We are conditioning on the event $\mathcal{E} = \{ \boldsymbol{M}\boldsymbol{V}_1 = \boldsymbol{0} \}$. We must determine what this event corresponds to in terms of $\boldsymbol{K}$.
    Substitute $\boldsymbol{M} = (\boldsymbol{U}^{-1})^T \boldsymbol{K} \boldsymbol{U}^{-1}$ into the condition:
    \[
    (\boldsymbol{U}^{-1})^T \boldsymbol{K} \boldsymbol{U}^{-1} \boldsymbol{V}_1 = \boldsymbol{0}.
    \]
    Since $(\boldsymbol{U}^T)^{-1}$ is invertible, we can multiply by $\boldsymbol{U}^T$ on the left, simplifying the condition to:
    \[
    \boldsymbol{K} (\boldsymbol{U}^{-1} \boldsymbol{V}_1) = \boldsymbol{0}.
    \]
    Also, from the definition of $\boldsymbol{U}$, we know that 
    $$
    \boldsymbol{U}^{-1} \boldsymbol{V}_1 = 
    \begin{bmatrix}
        I_{\ell} \\[4pt] 
        \boldsymbol{0}_{(d-\ell)\times\ell}
    \end{bmatrix},
    $$
    where $I_{\ell}$ is an identity matrix of size $\ell$.
    Thus, the condition $\boldsymbol{K} (\boldsymbol{U}^{-1} \boldsymbol{V}_1) = \boldsymbol{0}$ implies that the first $\ell$ columns of $\boldsymbol{K}$ are zero. Since $\boldsymbol{K}$ is symmetric, its first $\ell$ rows are also zero. Accordingly, define the event
    \[
    \mathcal{E}' \triangleq 
    \left\{
    \boldsymbol{K}=
    \begin{bmatrix}
    \boldsymbol{0}_{\ell\times \ell} & \boldsymbol{0}_{\ell\times (d-\ell)}\\
    \boldsymbol{0}_{(d-\ell)\times \ell} & \boldsymbol{G}
    \end{bmatrix}
    \;\middle|\;
    \boldsymbol{G}\in \mathcal{S}_{d-\ell}
    \right\},
    \]
    i.e., $\boldsymbol{K}$ has an all-zero $\ell\times d$ top block and an all-zero $d\times \ell$ left block, with an arbitrary symmetric $(d-\ell)\times(d-\ell)$ lower-right block. We conclude that the condition \(\boldsymbol{M}\boldsymbol{V}=0\) (i.e., event \(\mathcal{E}\)) implies that \(\boldsymbol{K} \in \mathcal{E}'\).
    
    Conversely, assume $\mathcal{E}'$ holds, i.e.,
    \[
    \boldsymbol{K}=
    \begin{bmatrix}
    \boldsymbol{0} & \boldsymbol{0}\\
    \boldsymbol{0} & \boldsymbol{G}
    \end{bmatrix}.
    \]
    Then
    \[
    \boldsymbol{K}\big(\boldsymbol{U}^{-1}\boldsymbol{V}_1\big)
    =\boldsymbol{K}
    \begin{bmatrix}
    \boldsymbol{I}_\ell\\[3pt]
    \boldsymbol{0}
    \end{bmatrix}
    =
    \begin{bmatrix}
    \boldsymbol{0}\\
    \boldsymbol{0}
    \end{bmatrix},
    \]
    and consequently
    \begin{align*}
        \boldsymbol{M}\boldsymbol{V}_1
    &=(\boldsymbol{U}^{-1})^{T}\boldsymbol{K}\boldsymbol{U}^{-1}\boldsymbol{V}_1 \\
    &=(\boldsymbol{U}^{-1})^{T}\big[\boldsymbol{K}(\boldsymbol{U}^{-1}\boldsymbol{V}_1)\big]\\
    &=(\boldsymbol{U}^{-1})^{T}\boldsymbol{0}
    =\boldsymbol{0}.
    \end{align*}
    Thus, $\mathcal{E}$ holds.
    Since both implications are established, $\mathcal{E}$ and $\mathcal{E}'$ are equivalent.

    It remains to characterize the conditional law of the lower-right block. Under $\mathcal{E}'$, the matrix $\boldsymbol{K}$ is uniquely determined by the choice of the symmetric block $\boldsymbol{G}\in\mathcal{S}_{d-\ell}$, via
\[
\boldsymbol{K}=
\begin{bmatrix}
\boldsymbol{0} & \boldsymbol{0}\\
\boldsymbol{0} & \boldsymbol{G}
\end{bmatrix}.
\]
Since $\boldsymbol{K}$ is uniform over $\mathcal{S}_d$, conditioning on $\mathcal{E}'$ makes $\boldsymbol{K}$ uniform over the subset
\[
\mathcal{S}_d(\mathcal{E}') \triangleq 
\left\{
\begin{bmatrix}
\boldsymbol{0} & \boldsymbol{0}\\
\boldsymbol{0} & \boldsymbol{G}
\end{bmatrix}
:\ \boldsymbol{G}\in\mathcal{S}_{d-\ell}
\right\},
\]
whose cardinality equals $|\mathcal{S}_{d-\ell}|$. Therefore, for any fixed $\boldsymbol{G}_0\in\mathcal{S}_{d-\ell}$,
\begin{align*}
\Pr\!\big(\boldsymbol{G}=\boldsymbol{G}_0 \,\big|\, \mathcal{E}\big)
&=\Pr\!\big(\boldsymbol{G}=\boldsymbol{G}_0 \,\big|\, \mathcal{E}'\big) \\[0.4em]
&=\frac{\Pr\!\left(\boldsymbol{K}=
\begin{bmatrix}
\boldsymbol{0} & \boldsymbol{0}\\
\boldsymbol{0} & \boldsymbol{G}_0
\end{bmatrix}\right)}{\Pr(\mathcal{E}')}\\[0.4em]
&=\frac{\frac{1}{|\mathcal{S}_d|}}{\frac{|\mathcal{S}_{d-\ell}|}{|\mathcal{S}_d|}}
=\frac{1}{|\mathcal{S}_{d-\ell}|}.
\end{align*}
Hence $\boldsymbol{G}$ is uniform over $\mathcal{S}_{d-\ell}$.

\end{proof}

\begin{lemma}\label{main_ind}
    Let $\boldsymbol{M} \in \mathbb{F}_q^{d \times d}$ be a symmetric random matrix with independent and uniformly distributed entries, let $\boldsymbol{v} \in \mathbb{F}_q^d$ be a Vandermonde vector, and let $\boldsymbol{V}_\epsilon \in \mathbb{F}_q^{d \times \ell}$ be a Vandermonde matrix with $\ell<d$, i.e., every column of $\boldsymbol{V}_\epsilon$ is a Vandermonde vector, where $\boldsymbol{v}$ is not in the column space of $\boldsymbol{V}_\epsilon$. Then, $[\boldsymbol{M}\boldsymbol{v}]_{1:d-\ell}$ is independent of $\boldsymbol{M}\boldsymbol{V}_\epsilon$.
\end{lemma}
\begin{proof}
    In order to prove the independence, we need to show that the following condition holds:
    \begin{equation}\label{eql321}
        H\left([\boldsymbol{M}\boldsymbol{v}]_{1:d-\ell} | \boldsymbol{M}\boldsymbol{V}_\epsilon \right) = H\left([\boldsymbol{M}\boldsymbol{v}]_{1:d-\ell} \right)
    \end{equation}
    
    Since $\boldsymbol{M}$ is drawn uniformly from the space of symmetric matrices over $\mathbb{F}_q$ and $\boldsymbol{v}$ is a fixed non-zero Vandermonde vector, the right-hand side term in \eqref{eql321} is uniformly distributed over $\mathbb{F}_q^{\,d-\ell}$. Indeed, each component
    \[
    [\boldsymbol{M}\boldsymbol{v}]_r
    = \sum_{j=1}^d [\boldsymbol{M}]_{r,j}\,[\boldsymbol{v}]_j,
    \]
    for $r = 1,\dots, d-\ell$, is a nontrivial linear combination of independent uniformly distributed entries of $\boldsymbol{M}$. Moreover, each such component contains a unique diagonal term $[\boldsymbol{M}]_{r,r}$ that does not appear in any other coordinate, ensuring that every coordinate is uniform over $\mathbb{F}_q$ and independent of the others. Hence,
    \begin{equation}\label{alpha_eq10}
        H\big([\boldsymbol{M}\boldsymbol{v}]_{1:d-\ell}\big) = d - \ell.
    \end{equation}

    Let 
    \[
    \mathcal{W} = \mathrm{span}\{\boldsymbol{v}_j : \boldsymbol{v}_j \in \text{columns of }  \boldsymbol{V}_\epsilon\} \subseteq \mathbb{F}_q^d .
    \] 
    Fix a realization $\boldsymbol{Y}_0 \in \mathbb{F}_q^{d\times \ell}$ of $\boldsymbol{M}\boldsymbol{V}_\epsilon$,
    choose a particular symmetric matrix $\boldsymbol{M}_0 \in \mathrm{Sym}_d(\mathbb{F}_q)$ satisfying
    \[
    \boldsymbol{M}_0 \boldsymbol{V}_\epsilon = \boldsymbol{Y}_0 .
    \]
    The full set of symmetric matrices consistent with this realization is the affine subspace
    \[
    \mathcal{S}_{\boldsymbol{Y}_0}
    = \big\{ \boldsymbol{M}\in\mathrm{Sym}_d(\mathbb{F}_q) : \boldsymbol{M}\boldsymbol{V}_\epsilon = \boldsymbol{Y}_0 \big\}
    = \boldsymbol{M}_0 + \mathcal{M}_{\mathcal{W}},
    \]
    where
    \[
    \mathcal{M}_{\mathcal{W}}
    = \big\{ \boldsymbol{M}' \in \mathrm{Sym}_d(\mathbb{F}_q) :
    \boldsymbol{M}'\boldsymbol{w} = \boldsymbol{0},\, \forall \boldsymbol{w} \in \mathcal{W} \big\}.
    \]
    Conditioning on $\boldsymbol{M}\boldsymbol{V}_\epsilon = \boldsymbol{Y}_0$ forces
    \[
    \boldsymbol{M} = \boldsymbol{M}_0 + \boldsymbol{M}' \quad \text{for some} \quad \boldsymbol{M}' \in \mathcal{M}_{\mathcal{W}},
    \]
    and the mapping $\boldsymbol{M} \mapsto \boldsymbol{M}' = \boldsymbol{M} - \boldsymbol{M}_0$ is a bijection between the affine space of solutions $\mathcal{S}_{\boldsymbol{Y}_0}$ and the linear space of annihilators $\mathcal{M}_{\mathcal{W}}$. Since $\boldsymbol{M}$ is drawn uniformly from $\mathcal{S}_{\boldsymbol{Y}_0}$, and the mapping is a simple translation, the resulting matrix $\boldsymbol{M}'$ is uniformly distributed over $\mathcal{M}_{\mathcal{W}}$.
    Therefore,
    \[
    H\!\left( [\boldsymbol{M}\boldsymbol{v}]_{1:d-\ell} \,\middle|\, \boldsymbol{M}\boldsymbol{V}_\epsilon = \boldsymbol{Y}_0 \right)
    =
    H\!\left( [(\boldsymbol{M}_0 + \boldsymbol{M}')\boldsymbol{v}]_{1:d-\ell} \right).
    \]
    Because $\boldsymbol{M}_0$ is deterministic under the conditioning, adding $\boldsymbol{M}_0$ only shifts the
    random vector $(\boldsymbol{M}_0 + \boldsymbol{M}')\boldsymbol{v}$ by a constant. Entropy over a finite alphabet is invariant under such translations, hence
    \[
    H\!\left( [(\boldsymbol{M}_0 + \boldsymbol{M}')\boldsymbol{v}]_{1:d-\ell} \right)
    =
    H\!\left( [\boldsymbol{M}'\boldsymbol{v}]_{1:d-\ell} \right),
    \]
    To remove the conditioning, apply the law of total entropy:
    \begin{equation}\label{eqfc1}
        \begin{aligned}
            &H\!\left( [\boldsymbol{M}\boldsymbol{v}]_{1:d-\ell} \,\middle|\, \boldsymbol{M}\boldsymbol{V}_\epsilon \right) \\
            &=
            \sum_{\boldsymbol{Y}_0}
            \Pr(\boldsymbol{M}\boldsymbol{V}_\epsilon = \boldsymbol{Y}_0)\,
            H\!\left( [\boldsymbol{M}\boldsymbol{v}]_{1:d-\ell} \,\middle|\, \boldsymbol{M}\boldsymbol{V}_\epsilon = \boldsymbol{Y}_0 \right)
            \\[0.4em]
            &=
            \sum_{\boldsymbol{Y}_0}
            \Pr(\boldsymbol{M}\boldsymbol{V}_\epsilon = \boldsymbol{Y}_0)\,
            H\!\left( [\boldsymbol{M}'\boldsymbol{v}]_{1:d-\ell} \right). \\[0.4em]
            &\overset{(a)}{=} H\!\left( [\boldsymbol{M}'\boldsymbol{v}]_{1:d-\ell} \right)
        \end{aligned}
    \end{equation}
    where (a) follows from the fact that each term in the sum is identical. So the weighted average equals that same value.

    Let $\boldsymbol{U} \in \mathbb{F}_q^d$ be a basis for $\mathbb{F}_q^d$, such that
    $$
    \boldsymbol{U} = [\boldsymbol{V}_\epsilon, \boldsymbol{v}, \boldsymbol{V}'],
    $$
    where $\boldsymbol{V}'\in\mathbb{F}_q^{d\times (d-\ell-1)}$ is a Vandermonde matrix whose columns lie outside the column space of $\boldsymbol{V}_\epsilon$ and $\boldsymbol{v}$, thereby completing the basis. Define
    $$
    \boldsymbol{K} = (\boldsymbol{U})^T \boldsymbol{M}' \boldsymbol{U}.
    $$
    Using Lemma \ref{Lemmfl}, for any symmetric $\boldsymbol{M}' \in \mathbb{F}_q^{d \times d}$ and any invertible $\boldsymbol{U} \in \mathbb{F}_q^{d \times d}$, and under the condition $\boldsymbol{M}'\boldsymbol{V}_\epsilon=\boldsymbol{0}_{d\times \ell}$, the matrix $\boldsymbol{K}$ takes the following block form
    \[
    \boldsymbol{K}=
    \begin{bmatrix}
        \boldsymbol{0}_{\ell\times \ell} & \boldsymbol{0}_{\ell\times (d-\ell)}\\[0.3em]
        \boldsymbol{0}_{(d-\ell)\times \ell} & \boldsymbol{\Gamma}
    \end{bmatrix},
    \]
    where $\boldsymbol{\Gamma}\in\mathbb{F}_q^{(d-\ell)\times(d-\ell)}$ is symmetric and its upper–triangular entries are independent and uniformly distributed over $\mathbb{F}_q$.

    Since $\boldsymbol{U}$ is invertible, $\boldsymbol{M}'$ can be written as
    \[
    \boldsymbol{M}' = (\boldsymbol{U}^T)^{-1} \,\boldsymbol{K}\, \boldsymbol{U}^{-1}.
    \]
    Hence, we have
    \begin{equation}\label{cb}
        \boldsymbol{M}' \boldsymbol{v} = (\boldsymbol{U}^{-1})^T \boldsymbol{K} \boldsymbol{U}^{-1} v \;,
    \end{equation}
    where $\boldsymbol{U}^{-1}\boldsymbol{v}$ equals the standard basis vector $\boldsymbol{e}_{\ell+1}$, i.e., the vector whose $(\ell+1)$-th entry is one and all other entries are zero. Thus, we have
    \begin{equation} \label{ke}
        \boldsymbol{K} \boldsymbol{e}_{\ell+1} = 
        \begin{bmatrix}
            \boldsymbol{0}_{\ell \times 1} \\
            [\boldsymbol{\Gamma}]_{:,1}
        \end{bmatrix}.
    \end{equation}
    Now can write the left-hand side term in \eqref{eql321} as
    \begin{equation}\label{lefteq}
        \begin{aligned}
            H\big([\boldsymbol{M}\boldsymbol{v}]_{1:d-\ell} |& \boldsymbol{M}\boldsymbol{V}_\epsilon \big) \overset{(a)}{=} H([\boldsymbol{M}' \boldsymbol{v}]_{1:d-\ell}) \\
            &\overset{(b)}{=} H\left(\left[(\boldsymbol{U}^{-1})^T \boldsymbol{K} \boldsymbol{U}^{-1} v\right]_{1:d-\ell}\right) \\
            &\overset{(c)}{=} H\left(\left[(\boldsymbol{U}^{-1})^T \begin{bmatrix}
                \boldsymbol{0}_{\ell \times 1} \\
                [\boldsymbol{\Gamma}]_{:,1}
            \end{bmatrix}\right]_{1:d-\ell}\right) \\
            &= H\left(\left[\left[(\boldsymbol{U}^{-1})^T\right]_{:,\ell + 1:d}  [\boldsymbol{\Gamma}]_{:,1}\right]_{1:d-\ell}\right) \\
            &= H\left(\left[(\boldsymbol{U}^{-1})^T\right]_{1:d-\ell,\ell + 1:d} [\boldsymbol{\Gamma}]_{:,1}\right) \\
            &\overset{(d)}{=} H\left( [\boldsymbol{\Gamma}]_{:,1}\right) \\
            &\overset{(e)}{=} d-\ell,
        \end{aligned}
    \end{equation}
    where (a) follows from \eqref{eqfc1}, (b) follows from \eqref{cb}, (d) follows from the fact that every square submatrix of a Vandermonde matrix is invertible, and (e) follows from the fact that the entropy of a vector with independent and uniformly distributed entries equals its dimension.
    
    As can be seen, \eqref{alpha_eq10} and \eqref{lefteq} are equal, which proves that the statement in \eqref{eql321} is correct and  $[\boldsymbol{M}\boldsymbol{v}]_{1:d-\ell}$ is independent of $\boldsymbol{M}\boldsymbol{V}_\epsilon$.

\end{proof}

\begin{lemma}\label{d-vertex_c}
Suppose that a node $v \in \mathcal{V}$ is $d$-vertex connected to a source $s \in \mathcal{S}$. Then, the source $s$ can securely transmit any message vector $\boldsymbol{x} \in \mathbb{F}_q^{d-\ell}$ to node $v$ at rate $d-\ell$, such that an eavesdropper observing up to $\ell < d$ nodes gains no information about the message, i.e.,
\[
H(\boldsymbol{x}\mid X_v)=0
\]
and
\[
\forall\,\beta\subseteq \mathcal{V}\setminus\{s,v\}\ \text{with}\ |\beta|\le \ell,\quad
I\!\left(\boldsymbol{x};\{X_u:u\in\beta\}\right)=0.
\]\end{lemma}

\begin{proof}
Since $v$ is $d$-vertex connected from $s$, there exist $d$ vertex-disjoint paths $P_1,\ldots,P_d$ from $s$ to $v$. 

We use Shamir's secret sharing scheme \cite{shamir1979share}, where the secret is the vector $\boldsymbol{x}\in\mathbb{F}_q^{d-\ell}$. Let $\boldsymbol{r} \in \mathbb{F}_q^{\ell}$ be a vector with independently and uniformly distributed elements over $\mathbb{F}_q$. Define the polynomial
\[
f(i) \triangleq 
\sum_{j=1}^{d-\ell} [\boldsymbol{x}]_{j} i^{j-1}
\;+\;
\sum_{j=1}^{\ell} [\boldsymbol{r}]_{j} i^{d-\ell+j-1}.
\]
Equivalently,
\begin{align*}
f(i)
&=
\big[[\boldsymbol{x}]_1,\ldots,[\boldsymbol{x}]_{d-\ell},
[\boldsymbol{r}]_1,\ldots,[\boldsymbol{r}]_{\ell}\big]
\begin{bmatrix}
i^{0}\\
i^{1}\\
\vdots\\
i^{d-1}
\end{bmatrix}\\
&=
\begin{bmatrix}
\boldsymbol{x}\\
\boldsymbol{r}\\
\end{bmatrix}^{T}
\begin{bmatrix}
i^{0}\\
i^{1}\\
\vdots\\
i^{d-1}
\end{bmatrix}.
\end{align*}
For each $i\in\{1,\ldots,d\}$, the source sends the share $f(i)$ along path $P_i$ to node $v$.

Node $v$ receives all $d$ shares and forms
\[
\boldsymbol{y}_v
=
\begin{bmatrix}
\boldsymbol{x}\\
\boldsymbol{r}\\
\end{bmatrix}^{T}\boldsymbol{V},
\]
where $\boldsymbol{V}\in\mathbb{F}_q^{d\times d}$ is the Vandermonde matrix whose $(j,i)$-th entry equals $i^{\,j-1}$. 
Since the evaluation points are distinct, $\boldsymbol{V}$ is invertible. Hence, node $v$ can recover the coefficient vector by computing $\boldsymbol{y}_v \boldsymbol{V}^{-1}$, and therefore reconstruct the secret $\boldsymbol{x}$.

Take any $\beta\subseteq \mathcal{V}\setminus\{s,v\}$ with $|\beta|\le \ell$. Because each node is on at most one path, the eavesdropper’s total information $\{f(i):i\in\beta\}$ is contained in a set of at most $|\beta|$ share values. Hence it suffices to prove that for any index set $O\subseteq\{1,\dots,d\}$ with $ |O| = \ell$,
\[
I(\boldsymbol{x};\boldsymbol{y}_O)=0,\qquad \boldsymbol{y}_O\triangleq (y_i)_{i\in O}.
\]
Write $\boldsymbol{y}_O$ as an affine function of $(\boldsymbol{x},\boldsymbol{r})$:
\[
\boldsymbol{y}_O = \boldsymbol{x}\boldsymbol{A}_O + \boldsymbol{r}\boldsymbol{B}_O,
\]
where $\boldsymbol{B}_O$ is an $\ell \times \ell$ submatrix of a Vandermonde matrix $\boldsymbol{V}$, and is therefore invertible.

Because $\boldsymbol{r}$ is uniform over $\mathbb{F}_q^\ell$ and $\boldsymbol{B}_O$ is an invertible, the random vector $\boldsymbol{r}\boldsymbol{B}_O$ is uniform over $\mathbb{F}_q^\ell$. Hence, $\boldsymbol{y}_O=\boldsymbol{x}\boldsymbol{A}_O+\boldsymbol{r}\boldsymbol{B}_O$ is also uniform over $\mathbb{F}_q^\ell$ and its distribution does not depend on $\boldsymbol{x}$. Equivalently,
\[
H(\boldsymbol{y}_O\mid \boldsymbol{x}) = \ell
\quad\text{and}\quad
H(\boldsymbol{y}_O)=\ell,
\]
so,
\[
I(\boldsymbol{x};\boldsymbol{y}_O)=H(\boldsymbol{y}_O)-H(\boldsymbol{y}_O\mid \boldsymbol{x})=0.
\]
Therefore, the desired secrecy holds for every $|\beta|\le \ell$.

The scheme delivers vector $\boldsymbol{x}\in\mathbb{F}_q^{d-\ell}$ to $v$ with perfect reconstruction and perfect secrecy against any $\ell<d$ observed nodes with a rate of $d-\ell$. 
\end{proof}

\begin{lemma}\label{KG_uniform} 
Let $\boldsymbol{A}\in\mathbb{F}_q^{(d-\ell)\times (n-m)}$ be a random matrix whose entries are independently and uniformly distributed over $\mathbb{F}_q$, and let $\boldsymbol{B}\in\mathbb{F}_q^{(n-m)\times (n-\ell)}$ be a fixed Vandermonde matrix with full column rank. If $\ell \ge m$, then the product $\boldsymbol{C} = \boldsymbol{A}\boldsymbol{B}$ is uniformly distributed over $\mathbb{F}_q^{(d-\ell)\times (n-\ell)}$.
\end{lemma}

\begin{proof}
Let $\mathbf{a}_i$ denote the $i$-th row of $\boldsymbol{A}$, for $1 \le i \le d-\ell$.
Since the entries of $\boldsymbol{A}$ are independently and uniformly distributed over $\mathbb{F}_q$, each row $\mathbf{a}_i$ is uniformly distributed over $\mathbb{F}_q^{1\times (n-m)}$, and the rows $\mathbf{a}_1,\dots,\mathbf{a}_{d-\ell}$ are mutually independent.

Let $\mathbf{c}_i$ denote the $i$-th row of the product matrix $\boldsymbol{C}=\boldsymbol{A}\boldsymbol{B}$. By matrix multiplication,
\[
\mathbf{c}_i = \mathbf{a}_i \boldsymbol{B}.
\]
Since each $\mathbf{c}_i$ depends only on the corresponding row $\mathbf{a}_i$ and the deterministic matrix $\boldsymbol{B}$, the rows $\mathbf{c}_1,\dots,\mathbf{c}_{d-\ell}$ remain mutually independent. Therefore, to show that $\boldsymbol{C}$ is uniformly distributed over $\mathbb{F}_q^{(d-\ell)\times (n-\ell)}$, it suffices to prove that, for any fixed $i\in[d-\ell]$, the row vector $\mathbf{c}_i$ is uniformly distributed over $\mathbb{F}_q^{1\times (n-\ell)}$.

We can rewrite $\mathbf{c}_i$ as
\[
\mathbf{c}_i
=
[\mathbf{a}_i]_{1,1:n-\ell}\,[\boldsymbol{B}]_{1:n-\ell,:}
+
[\mathbf{a}_i]_{1,n-\ell+1:n-m}\,[\boldsymbol{B}]_{n-\ell+1:n-m,:}.
\]
The subvector $[\mathbf{a}_i]_{1,1:n-\ell}$ is uniformly distributed over $\mathbb{F}_q^{1\times (n-\ell)}$. Since
$[\boldsymbol{B}]_{1:n-\ell,:}\in\mathbb{F}_q^{(n-\ell)\times (n-\ell)}$
is a square submatrix of a Vandermonde matrix, and is therefore invertible,
\begin{equation}\label{sumterm}
    [\mathbf{a}_i]_{1,1:n-\ell}[\boldsymbol{B}]_{1:n-\ell,:},
\end{equation}
is also uniformly distributed over $\mathbb{F}_q^{1\times (n-\ell)}$.
Since \eqref{sumterm} is uniformly distributed over $\mathbb{F}_q^{1\times (n-\ell)}$, the vector $\mathbf{c}_i$ is also uniformly distributed over $\mathbb{F}_q^{1\times (n-\ell)}$, because adding another term does not alter uniformity.

Since this holds for all $i \in [d-\ell]$ and the rows $\mathbf{c}_i$ are mutually independent, the entire matrix $\boldsymbol C = \boldsymbol A \boldsymbol B$ is uniformly distributed over $\mathbb{F}_q^{(d-\ell)\times (n-\ell)}$.

\end{proof}

\section{$d$-vertex connected networks}\label{FCsec}
In this section, we address network instances in which all nodes are $d$-vertex connected from the source(s). We obtain a number of results.
First, we show that when an eavesdropper can observe any set of $\ell<d$ nodes, there exist network instances for which the maximum key rate is $d-\ell$.
This establishes that the loss of $\ell$ units of rate due to adversarial observation is, in general, unavoidable even under strong connectivity assumptions.
Next, we show that this upper bound is tight. In particular, for any single-source instance $\mathcal{I}$ of the secure multiple key-cast problem in which all nodes are $d$-vertex connected from the source, there exists a coding scheme that achieves a key rate of $d-\ell$ while providing perfect secrecy against any eavesdropper observing up to $\ell$ nodes. 
We then extend this result to the multiple-source setting. For instances with a set of sources $\mathcal{S}=\{s_1,\dots,s_{|\mathcal{S}|}\}$, we show that if every node in the network is $d$-vertex connected from each source $s \in \mathcal{S}$, a key rate of $\frac{(d-\ell)(|\mathcal{S}|-\ell)}{|\mathcal{S}|}$ is achievable with perfect secrecy against any eavesdropper observing up to $\ell$ non-source nodes and up to $x$ source nodes, provided that $d>\ell$ and $|\mathcal{S}|>x$. The proof follows from a constructive coding scheme similar to the single-source case.

\begin{lemma}[Converse on secure key-cast rate]\label{L:converse}
    There exists a network instance $\mathcal{I}=\bigl(\mathcal{V}, \mathcal{E}, \mathcal{B}, \mathcal{S}=\{s\}, \mathcal{T}=\{t\}\bigr)$ with 
    \(
    \mathcal{B}=
    \bigl\{
    \beta \subseteq \mathcal{V} : \beta \cap \{s\}=\varnothing,\ |\beta|\le \ell
    \bigr\},\)
such that every node is $d$-vertex connected from the source $s$, for which the optimal (i.e., maximum achievable) secure key rate between $s$ and $t$ is at most $d-\ell$.
\end{lemma}

\begin{proof}
Consider the instance $\mathcal{I}=(\mathcal{V},\mathcal{E},\mathcal{B},\mathcal{S},\mathcal{T})$ illustrated in Fig.~\ref{fig:placeholder}. Let
\[
\mathcal{V}=\{s,t,v_1,\dots,v_d\},\qquad \mathcal{S}=\{s\},\qquad \mathcal{T}=\{t\},
\]
and
\begin{align*}
    \mathcal{E}
=\big\{(s,v_i)^{(j)} : i\in\{1,\dots,d\}, j\in\{1,\dots,d\}\big\} \\
\cup
\{(v_i,t): i\in\{1,\dots,d\}\},
\end{align*}
where \((s,v_i)^{(j)}\) denotes the \(j\)-th parallel edge from \(s\) to \(v_i\), and each edge \((v_i,t)\) has unit capacity (one symbol from \(\mathbb{F}_q\) per channel use).

The $d$ paths $s\to v_i\to t$ are vertex-disjoint; hence $t$ is $d$-vertex-connected to $s$. Moreover, each node $v_i$ is connected to $s$ by $d$ parallel edges, implying that $v_i$ is $d$-vertex-connected to $s$.

Fix blocklength $n$ and any key-code producing a key $K$ such that
\[
H(K\mid X_t)=0,\qquad \forall\,\beta\in\mathcal{B},\ I\!\left(K;\{X_v:v\in\beta\}\right)=0,
\]
where $X_v$ is the information received by node $v$ on its incoming edges.

Fix any $\beta\in\mathcal{B}$, let $\bar\beta=\{v_1,\dots,v_d\}\setminus\beta$, so $|\bar\beta|=d-\ell$.
For each $i$, let $Y_i^n$ denote the length-$n$ message sent on edge $(v_i,t)$. Then
\[
X_t=(Y_1^n,\dots,Y_d^n).
\]
Moreover, for each $v_i\in\beta$, the eavesdropper observes node $v_i$ and hence can reconstruct all transmissions on edges incident to $v_i$, in particular $Y_i^n$. Therefore, conditioned on the eavesdropper’s observation $\{X_v:v\in\beta\}$, the tuple $(Y_i^n:v_i\in\beta)$ is known, and the only remaining uncertainty in $X_t$ comes from $(Y_i^n:v_i\in\bar\beta)$.

Using secrecy and correctness,
\begin{align*}
H(K)
&= H\!\left(K \mid \{X_v:v\in\beta\}\right) \\
&\le H\!\left(X_t \mid \{X_v:v\in\beta\}\right) \\
&= H\!\left((Y_i^n:v_i\in\bar\beta)\mid \{X_v:v\in\beta\}\right) \\
&\le H\!\left((Y_i^n:v_i\in\bar\beta)\right) \\
&\le \sum_{v_i\in\bar\beta} H(Y_i^n)
\ \le\ (d-\ell)n,
\end{align*}
where the last inequality uses that each edge $(v_i,t)$ has unit capacity, so $Y_i^n$ contains at most $n$ $\mathbb{F}_q$-symbols.

Thus $H(K)\le (d-\ell)n$, and the key rate $R=H(K)/n$ satisfies $R\le d-\ell$. Since $\beta$ was arbitrary in $\mathcal{B}$, the optimal secure key rate is at most $d-\ell$.
\end{proof}

\begin{figure}
    \centering
    \includegraphics[width=0.4\linewidth]{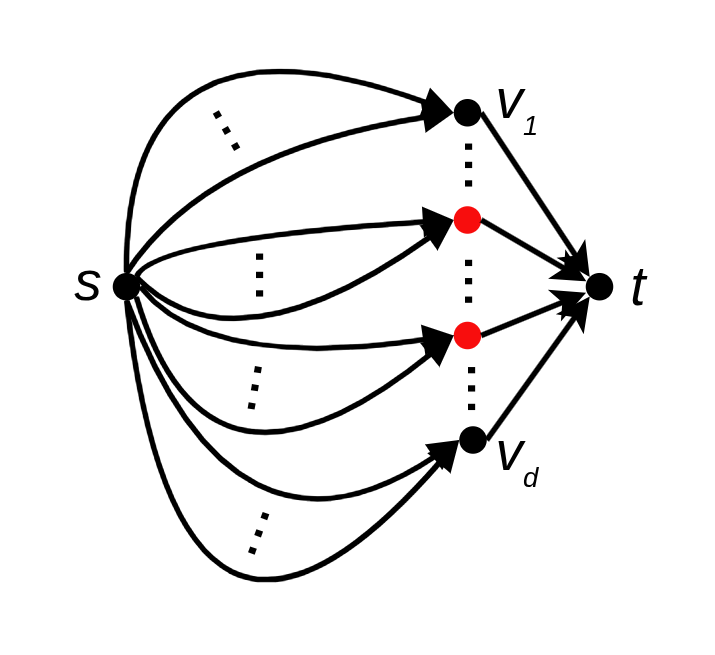}
    \vspace{-5mm}
     \caption{Example network with a source, a terminal, and $d$ intermediate nodes. An eavesdropper observes any $\ell$ intermediate nodes (an example is shown in red), leaving $d-\ell$ secure paths that can be used for secure key-cast. 
     }
    \label{fig:placeholder}
\end{figure}

\begin{theorem}[{Achievability: single-source secure multiple key-cast}]
\label{T1}
    For an instance
    $\mathcal{I}=\left(\mathcal{V},\mathcal{E},\mathcal{B},\mathcal{S}=\{s\},\mathcal{T}=\{\mathcal{T}_i\}_{i=1}^m\right)$
    of the Secure Multiple Key-Cast problem with
    \(\mathcal{B}=
    \bigl\{
    \beta \subseteq \mathcal{V}\setminus \{s\} :  |\beta|\le \ell
    \bigr\},\)
 if all nodes in the network are \(d\)-vertex connected from the source \(s\), then there exists a coding scheme that achieves key rate \(d-\ell\).
\end{theorem}

\begin{proof}
    Consider a network where all nodes are $d$-vertex connected from the source. Each non-terminal node $v_i \in \mathcal{V}$ is assigned a Vandermonde vector indexed by $i$, and for each terminal set  $\mathcal{T}_i \subseteq \mathcal{T}$, all terminal nodes $v_t \in \mathcal{T}_i$ share the same Vandermonde vector, denoted by $\boldsymbol{v}_{\mathcal{T}_i}$.
    
    Define the key for terminal set $\mathcal{T}_i$ as
    \begin{equation}\label{key-defT1}
        \boldsymbol{k}_{\mathcal{T}_i} = [\boldsymbol{M} \boldsymbol{v}_{\mathcal{T}_i} ]_{1:d-\ell},
    \end{equation}
    where $[\boldsymbol{M} \boldsymbol{v}_{\mathcal{T}_i}]_{1:d-\ell}$ is the first $d-\ell$ entries of $\boldsymbol{M} \boldsymbol{v}_{\mathcal{T}_i}$, and $\boldsymbol{M} \in \mathbb{F}_q^{d \times d}$ is an independent symmetric random matrix of size $d \times d$, with each entry in its upper triangular part is independently and uniformly chosen from $\mathbb{F}_q$.
    
    We start by defining our encoding scheme.
    
    \subsection*{Encoding:}

    \paragraph*{Step 1} For each node $j \in \mathcal{V}\setminus\{s\}$, let $c_j \le d$ be the number of edges from $s$ to $j$. The source node sends $c_j$ symbols using an independent Vandermonde matrix $\boldsymbol{V}_{s\to j}\in\mathbb{F}_q^{d\times c_j}$, constructed from distinct Vandermonde vectors unused elsewhere in the network, as the row vector:
    \begin{equation}
        \boldsymbol{\tau}_{s \to j} = \boldsymbol{v}_j^T \boldsymbol{M} \boldsymbol{V}_{s\to j}.
    \end{equation}
    Note that the first non-source node in the topological order have no non-source parents, meaning it must have $c_j = d$ edges from the source to satisfy $d$-vertex connectivity. For such node $j$, $\boldsymbol{V}_{s\to j}$ is a $d \times d$ invertible matrix, allowing it to directly recover 
    $$\boldsymbol{s}_j = (\boldsymbol{\tau}_{s \to j} \boldsymbol{V}_{s\to j}^{-1})^T = \boldsymbol{M}\boldsymbol{v}_j.$$
    We refer to $\boldsymbol{s}_j$ as the \emph{share} of node $j$.

\paragraph*{Step 2} 
    By induction, let \( j \in \mathcal{V} \setminus \{s\} \) denote the next node in the topological order such that every preceding node \( j' \) has already obtained its share, given by
    \[
    \boldsymbol{s}_{j'} = \boldsymbol{M} \boldsymbol{v}_{j'}.
    \]
    Node $j$ receives the remaining $d-c_j$ symbols from its non-source parents $j_p \in \mathcal{P}(j)$, where each parent sends:
    $$
    \tau_{j_p \to j} = \boldsymbol{v}_{j}^T \boldsymbol{M} \boldsymbol{v}_{{j}_p}. 
    $$
    Node $j$ stacks all received symbols to form the $1 \times d$ row vector:
    $$
    \boldsymbol{T}_{j} = \boldsymbol{v}_{j}^T \boldsymbol{M} \boldsymbol{V}_{j},
    $$
    where $\boldsymbol{V}_j \in \mathbb{F}_q^{d \times d}$ is the matrix whose first $c_j$ columns are $\boldsymbol{V}_{s\to j}$ and whose remaining columns are $\{\boldsymbol{v}_{j_p}\}_{j_p \in \mathcal{P}(j)}$.
    Due to $d$-vertex connectivity and distinct Vandermonde indices, the columns of $\boldsymbol{V}_{j}$ are linearly independent. Since $\boldsymbol{V}_{j}$ is invertible, node $j$ can recover its share as
    \begin{equation} \label{eq_case1}
        \boldsymbol{s}_{j} = (\boldsymbol{T}_{j}\boldsymbol{V}_{j}^{-1})^T = (\boldsymbol{v}_{j}^T \boldsymbol{M} \boldsymbol{V}_{j} \boldsymbol{V}_{j}^{-1})^T = \boldsymbol{M}\boldsymbol{v}_{j} .
    \end{equation}

\paragraph*{Step 3} 
    Let $t \in {\mathcal{T}_i}$ be a terminal node. Following the exact same procedure as Step 1 and Step 2, $t$ collects $c_t$ symbols from the source and $d-c_t$ symbols from its non-source parents $t_p \in \mathcal{P}(t)$, which send
    \[
    \tau_{t_p \to t} = \boldsymbol{v}_{\mathcal{T}_i}^T \boldsymbol{M} \boldsymbol{v}_{t_p}.
    \]
    Using all received symbols, the terminal node $t$ reconstructs the $1 \times d$ row vector 
    \begin{equation}
        \boldsymbol{T}_{t} = \boldsymbol{v}_{\mathcal{T}_i}^T \boldsymbol{M} \boldsymbol{V}_{t},
    \end{equation}
    where $\boldsymbol{V}_t$ is the invertible $d \times d$ matrix whose columns are formed by $\boldsymbol{V}_{s\to t}$ and $\{\boldsymbol{v}_{t_p}\}_{t_p \in \mathcal{P}(t)}$. Terminal $t$ recovers its share as
    \begin{equation} 
        \boldsymbol{s}_{t} = (\boldsymbol{T}_{t} \boldsymbol{V}_{t} ^{-1})^T = (\boldsymbol{v}_{\mathcal{T}_i}^T \boldsymbol{M} \boldsymbol{V}_{t}  \boldsymbol{V}_{t} ^{-1})^T = \boldsymbol{M} \boldsymbol{v}_{\mathcal{T}_i}.
    \end{equation}

    \subsection*{Decoding:}
    
    \paragraph*{Step 4} Using the recovered vector $\boldsymbol{s}_t$, the terminal node $t \in \mathcal{T}_i$ recovers the key as  
    \begin{equation} \label{eq_ket}
        \boldsymbol{k}_{\mathcal{T}_i} 
        = [\boldsymbol{s}_t]_{1:d-\ell} = [\boldsymbol{M} \boldsymbol{v}_{\mathcal{T}_i}]_{1:d-\ell},
    \end{equation}
    which is the defined key in \eqref{key-defT1} for the terminal set $\mathcal{T}_i \subseteq \mathcal{T}$.

    Algorithm~\ref{alg:keycast_compact} summarizes the resulting key-code for securely delivering $\boldsymbol{k}_{\mathcal{T}_i}$ to every terminal node $t \in \mathcal{T}_i$.   
    
    \subsection*{Key-Rate:}
    Each terminal node $t \in \mathcal{T}_i \subseteq \mathcal{T}$ recovers the key $\boldsymbol{k}_{\mathcal{T}_i}$ consists of $d-\ell$ components of $\boldsymbol{M} \boldsymbol{v}_{\mathcal{T}_i}$. Since the proposed scheme delivers $d-\ell$ key symbols while requiring at most a single transmission on each edge in $\mathcal{G}$, the key rate is given by
    $$
    R \;=\; d-\ell .
    $$
    
    \subsection*{Security Proof:}
    Every node $j \in \mathcal{V}$ in $\mathcal{G}$ has access to its assigned share
    \[
    \boldsymbol{s}_j = \boldsymbol{M}\boldsymbol{v}_j.
    \]
    Note that node $j$ actually holds gathered symbols in $\boldsymbol{T}_j$, which are computable from $\boldsymbol{s}_j$, and thus, equivalent to holding $\boldsymbol{s}_j$.
    Likewise, each terminal node $t \in \mathcal{T}_i$ obtains
    \[
    \boldsymbol{s}_{t} = \boldsymbol{M}\boldsymbol{v}_{\mathcal{T}_i}.
    \]
    
    We now show that for any $\beta \in \mathcal{B}$, the eavesdropper’s observation reveals no information about the keys not intended for the eavesdropper. Specifically, for any terminal set $\mathcal{T}_i \subseteq \mathcal{T}$, if $\beta \cap \mathcal{T}_i = \emptyset$, then we have
    \begin{equation}\label{mu-eq8}
        \begin{aligned}
            I(\boldsymbol{k}_{\mathcal{T}_i};\, &\{X_v:v\in\beta\}) = I\!\left([\boldsymbol{M}\boldsymbol{v}_{\mathcal{T}_i}]_{1:d-\ell};\, \boldsymbol{M}\boldsymbol{V}_\epsilon\right) \\
            &= H([\boldsymbol{M}\boldsymbol{v}_{\mathcal{T}_i}]_{1:d-\ell}) - H\!\left([\boldsymbol{M}\boldsymbol{v}_{\mathcal{T}_i}]_{1:d-\ell}\,\middle|\,\boldsymbol{M}\boldsymbol{V}_\epsilon\right) \\
            &\overset{(a)}{=} H([\boldsymbol{M}\boldsymbol{v}_{\mathcal{T}_i}]_{1:d-\ell}) - H([\boldsymbol{M}\boldsymbol{v}_{\mathcal{T}_i}]_{1:d-\ell}) \\
            &= 0
        \end{aligned}
    \end{equation}
    where $\boldsymbol{V}_\epsilon \in \mathbb{F}_q^{d \times \ell}$ is the Vandermonde matrix formed by the eavesdropped nodes, with columns $\{ \boldsymbol{v}_{j} \}_{j \in \beta}$, and (a) follows from Lemma~\ref{main_ind} and the fact that the vector $\boldsymbol{v}_{\mathcal{T}_i}$ is independent of $\boldsymbol{V}_\epsilon$. Hence, the eavesdropper obtains no information about the keys $\boldsymbol{k}_{\mathcal{T}_i}$, as desired.
\end{proof}

\begin{algorithm}[t]
\caption{Multiple Key-Cast in $d$-connected networks}
\label{alg:keycast_compact}
\begin{algorithmic}[1]
\State Source node $s$ generates a symmetric random matrix $\boldsymbol{M}\in\mathbb{F}_q^{d\times d}$, where each entry in the upper triangular part is chosen independently and uniformly from $\mathbb{F}_q$.
For $i \in \{1,\dots, m\}$, define keys  $\boldsymbol{k}_{\mathcal{T}_i}= [\boldsymbol{M}\boldsymbol{v}_{\mathcal{T}_i}]_{1:d-\ell}.$

\ForAll{$j\in\mathcal{N}(s)$}
    \State $s$ sends $\boldsymbol{v}_j^T \boldsymbol{M}\boldsymbol{V}_{s\to j}$ to node $j$
\EndFor

\ForAll{nodes $j\in\mathcal{V}\setminus\{s\}$ in topological order}
    \State Each non-source parent node $j_p\in\mathcal{P}(j)$ sends $\boldsymbol{v}_j^T\boldsymbol{M}\boldsymbol{v}_{j_p}$ to  $j$
    \State Node $j$ concatenates received symbols to form $\boldsymbol{T}_j=\boldsymbol{v}_j^T\boldsymbol{M}\boldsymbol{V}_j$ and recovers share
    \[
    \boldsymbol{s}_j
    =
    \bigl(\boldsymbol{T}_j\boldsymbol{V}_j^{-1}\bigr)^T
    =
    \bigl(\boldsymbol{v}_j^T\boldsymbol{M}\boldsymbol{V}_j\boldsymbol{V}_j^{-1}\bigr)^T
    =
    \boldsymbol{M}\boldsymbol{v}_j
    \]
\EndFor

\ForAll{$t\in\mathcal{T}_i$}
    \State Terminal $t$ outputs
    \(
    \boldsymbol{k}_{\mathcal{T}_i}
    =
    [\boldsymbol{s}_t]_{1:d-\ell}
    =
    [\boldsymbol{M}\boldsymbol{v}_{\mathcal{T}_i}]_{1:d-\ell}
    \)
\EndFor
\end{algorithmic}
\end{algorithm}

{\begin{remark}\label{R0}
    The construction in Algorithm~\ref{alg:keycast_compact} used in the proof of Theorem~\ref{T1} does not require the full network $\mathcal{G}$ itself to have every node $d$-vertex connected from the source. For a general network, it suffices that there exists a subgraph $\mathcal{G}' \subseteq \mathcal{G}$, containing the source, all terminal nodes, and possibly only a subset of the intermediate nodes, within which every terminal node and every intermediate node is $d$-vertex connected from the source. Whenever such a subgraph $\mathcal{G}'$ exists, the coding scheme summarized in Algorithm 1 can be run directly on $\mathcal{G}'$, using only its nodes and edges, and achieves a secure key rate of $d-\ell$. This observation extends the applicability of Theorem~\ref{T1} beyond networks that are globally $d$-vertex connected, to any network that merely admits such a $d$-vertex-connected subgraph.
\end{remark}}

\begin{corollary}[Achievability: multiple-source secure multiple key-cast]\label{C1}
    Consider an instance $\mathcal{I}=(\mathcal{V},\mathcal{E},\mathcal{B},\mathcal{S},\mathcal{T}=\{\mathcal{T}_i\}_{i=1}^m)$ of the secure multiple key-cast problem with a set of sources $\mathcal{S}=\{s_1,\dots,s_{|\mathcal{S}|}\}$.
    For any integer $x < |\mathcal{S}|$, let
    \(
    \mathcal{B}=
    \bigl\{
    \beta \subseteq \mathcal{V}: |\beta \setminus \mathcal{S}|\le \ell,\ |\beta \cap \mathcal{S}| \le x
    \bigr\}.
    \)
    If every node in the network is $d$-vertex connected from each source $s \in \mathcal{S}$, then there exists a coding scheme that achieves a key rate of
    \(
    \frac{(d-\ell)(|\mathcal{S}|-x)}{|\mathcal{S}|}.
    \)
    In particular, $x$ denotes the maximum number of source nodes that can be included in the eavesdropped set, and $x=0$ prohibits eavesdropping on any source node.
\end{corollary}

\begin{proof}
Consider a network in which all nodes are $d$-vertex connected from all source nodes $s \in \mathcal{S}$. 
For each source $s_j$, independently generate a symmetric random matrix $\boldsymbol{M}^{(j)} \in \mathbb{F}_q^{d\times d}$ and execute the coding scheme of Algorithm~\ref{alg:keycast_compact} to securely disseminate 
$\boldsymbol{k}_{\mathcal{T}_i}^{(j)} = [\boldsymbol{M}^{(j)} \boldsymbol{v}_{\mathcal{T}_i}]_{1:d-\ell}$
to all terminals $t \in \mathcal{T}_i$, where the randomness $\{\boldsymbol{M}^{(j)}\}_{j=1}^{|\mathcal{S}|}$ is mutually independent.

Define the $(d-\ell) \times |\mathcal{S}|$ matrix
\[
\tilde{\boldsymbol{K}}_{\mathcal{T}_i}
\triangleq
\big[ \boldsymbol{k}_{\mathcal{T}_i}^{(1)}, \dots , \boldsymbol{k}_{\mathcal{T}_i}^{(|\mathcal{S}|)} \big].
\]
Then, $(d-\ell) \times (|\mathcal{S}| - x)$ key matrix for terminal set $\mathcal{T}_i$ is defined as
\begin{equation}\label{keyCor11}
\boldsymbol{K}_{\mathcal{T}_i}
= \tilde{\boldsymbol{K}}_{\mathcal{T}_i}\boldsymbol{G}.
\end{equation}
where $\boldsymbol{G}$ is a fixed $|\mathcal{S}| \times (|\mathcal{S}|-x)$ Vandermonde matrix.

Using the same encoding and decoding procedures established in Algorithm~\ref{alg:keycast_compact}, each terminal recovers all values $\{\boldsymbol{k}_{\mathcal{T}_i}^{(u)} : u \in \mathcal{S}\}$ and constructs its final key as \eqref{keyCor11}.

\subsection*{Key-Rate:}
Each terminal $t \in \mathcal{T}_i$ combines the keys received from all $|\mathcal{S}|$ source nodes. Since the eavesdropper can observe at most $x$ source nodes, at least $|\mathcal{S}|-x$ of these keys remain secure, each containing $d-\ell$ components.
Then, it can obtain $(d-\ell)(|\mathcal{S}|-x)$ key symbols by combining the $|\mathcal{S}|$ vectors of size $d-\ell$ that it receives from the $|\mathcal{S}|$ source nodes using the matrix $\boldsymbol{G}$.
Since the scheme delivers $(d-\ell)(|\mathcal{S}|-x)$ key symbols while requiring at most one transmission from each of the $|\mathcal{S}|$ sources over any edge of $\mathcal{G}$, the resulting key rate is
$$
R \;=\; \frac{(d-\ell)(|\mathcal{S}|-x)}{|\mathcal{S}|}.
$$

\subsection*{Security Proof:}
For any $\beta \in \mathcal{B}$, first, we assume that no source node is included in the eavesdropper's observation, i.e., $x = 0 $.
We then consider the more general case in which the eavesdropper's observation may include source nodes, i.e.,
$x \neq 0$.

\paragraph{No source node in $\beta \in \mathcal{B}$ $(x=0)$} If there exists no source node in the eavesdropper's observation, then for any $\beta \in \mathcal{B}$,
\begin{equation}\label{NSNob1}
    \begin{aligned}
        I(\boldsymbol{K}_{\mathcal{T}_i}; \{X_v:&v\in\beta\}) \overset{(a)}{\le} I\left( \big( \boldsymbol{k}_{\mathcal{T}_i}^{(1)}, \dots , \boldsymbol{k}_{\mathcal{T}_i}^{(|\mathcal{S}|)} \big); \{X_v:v\in\beta\}\right) \\
        &\overset{(b)}{=} \summ{j}{1}{|\mathcal{S}|} I\big(\boldsymbol{k}_{\mathcal{T}_i}^{(j)}; \{X_v:v\in\beta\}|\boldsymbol{k}_{\mathcal{T}_i}^{(1)}, \dots, \boldsymbol{k}_{\mathcal{T}_i}^{(j-1)}\big) \\
        &\overset{(c)}{=} \summ{j}{1}{|\mathcal{S}|} I\big(\boldsymbol{k}_{\mathcal{T}_i}^{(j)}; \big(\{X_v^{(j')}:v\in\beta\}: j' \in [|\mathcal{S}|]\big) \\
        &\qquad \qquad \qquad\qquad\qquad |\boldsymbol{k}_{\mathcal{T}_i}^{(1)}, \dots, \boldsymbol{k}_{\mathcal{T}_i}^{(j-1)}\big) \\
        & \overset{(d)}{=} \summ{j}{1}{n} I\big(\boldsymbol{k}_{\mathcal{T}_i}^{(j)}; \{X_v^{(j)}:v\in\beta\}\big) \\
        &\overset{(e)}{=} 0,
    \end{aligned}
\end{equation}
where (a) follows from \eqref{keyCor11} and the fact that $\boldsymbol{G}$ is a fixed matrix, 
(b) follows from the chain rule,
and (c) follows from the construction in which the transmissions associated with different source nodes are kept separate throughout the network, i.e., symbols originating from distinct sources are never combined into a single transmitted symbol. Consequently, the eavesdropper's observation $\{X_v:v\in\beta\}$ can be written a collection of disjoint sets
\[
\{X_v:v\in\beta\} = \left\{\{X_v^{(1)}:v\in\beta\}, \dots , \{X_v^{(|\mathcal{S}|)}:v\in\beta\} \right\} ,
\]
where $\{X_v^{(j)}:v\in\beta\}$ depends only on transmissions originating from the $j^{th}$ source node. 
Step (d) follows from the fact that for each $j$,
\begin{equation}\label{eq:O_decomp_indep}
\begin{aligned}
    \big(\boldsymbol{k}_{\mathcal{T}_i}^{(j)}&,\{X_v^{(j)}:v\in\beta\}\big) \\
    &\ \perp\!\!\!\perp
    \Big(\boldsymbol{k}_{\mathcal{T}_i}^{(1)},\dots,\boldsymbol{k}_{\mathcal{T}_i}^{(j-1)}, \\
    &\qquad \{X_v^{(1)}:v\in\beta\},\dots, \{X_v^{(j-1)}:v\in\beta\},\\
    &\qquad \{X_v^{(j+1)}:v\in\beta\},\dots, \{X_v^{(|\mathcal{S}|)}:v\in\beta\}\Big),
\end{aligned}
\end{equation}
since all sources use independent randomness and their transmissions remain separated.
Finally, step (e) follows from the security proof of Theorem \ref{T1}.

\paragraph{Some source nodes are in $\beta \in \mathcal{B}$ $(x \neq 0)$}
Previously we had shown that if there does not exist any source node in the eavesdropper's observation, security is guaranteed. Now assume there exist some source nodes in $\beta \in \mathcal{B}$. Then, the eavesdropper's observation $\beta$ can be written as
\[
\beta= \big\{ \beta_{\epsilon \setminus s}, \beta_s \big\},
\]
where $\beta_s$ and $\beta_{\epsilon \setminus s}$ denote the eavesdropper's observations of the source nodes and non-source nodes, respectively, with
$|\beta_s| = x$ and $|\beta_{\epsilon \setminus s}| = \ell$.

Now we show that the eavesdropper's observation reveals no information about the key. Let $\mathcal{J} \triangleq \{ j : \mathcal{T}_j \subseteq \mathcal{T}\}$ denote the index set of terminal sets. Then, for any $i \in \mathcal{J}$ satisfying $\beta \cap \mathcal{T}_i = \emptyset$, we have
\begin{equation}
    \begin{aligned}
        I\big(\boldsymbol{K}_{\mathcal{T}_i};& \{X_v:v\in\beta\}\big) = I\big(\tilde{\boldsymbol{K}}_{\mathcal{T}_i}\boldsymbol{G};\; \{X_v:v\in\beta\} \big) \\
        & = I\big(\tilde{\boldsymbol{K}}_{\mathcal{T}_i}\boldsymbol{G};\; \{X_v:v\in\beta_{\epsilon \setminus s}\}, \{X_s:s\in\beta_s\} \big) \\
        & \overset{(a)}{=} I\big(\tilde{\boldsymbol{K}}_{\mathcal{T}_i}\boldsymbol{G};\;\{X_v:v\in\beta_{\epsilon \setminus s}\}, (\boldsymbol{M}^{(j)}:s_j \in  \beta_s) \big) \\
        & = H\big( \tilde{\boldsymbol{K}}_{\mathcal{T}_i}\boldsymbol{G} \big) - H\big(\tilde{\boldsymbol{K}}_{\mathcal{T}_i}\boldsymbol{G}\\
        & \qquad \qquad \quad \; \;|\{X_v:v\in\beta_{\epsilon \setminus s}\}, (\boldsymbol{M}^{(j)}:s_j \in  \beta_s) \big) \\
        & {=} H\big( \tilde{\boldsymbol{K}}_{\mathcal{T}_i}\boldsymbol{G} \big) - H\big( {\tilde{\boldsymbol{K}}_{\mathcal{T}_i}}^{({\Omega})}\boldsymbol{G}^{({\Omega})} + {\tilde{\boldsymbol{K}}_{\mathcal{T}_i}}^{(\bar{\Omega})}\boldsymbol{G}^{(\bar{\Omega})}\\
        & \qquad \qquad \quad \; \;|\{X_v:v\in\beta_{\epsilon \setminus s}\}, (\boldsymbol{M}^{(j)}:s_j \in  \beta_s) \big) \\
        & \overset{(b)}{=} H\big( \tilde{\boldsymbol{K}}_{\mathcal{T}_i}\boldsymbol{G} \big) - H\big( {\tilde{\boldsymbol{K}}_{\mathcal{T}_i}}^{(\bar{\Omega})}\boldsymbol{G}^{(\bar{\Omega})}\\
        & \qquad \qquad \quad \; \;|\{X_v:v\in\beta_{\epsilon \setminus s}\}, (\boldsymbol{M}^{(j)}:s_j \in  \beta_s) \big) \\
        & \overset{(c)}{=} H\big( \tilde{\boldsymbol{K}}_{\mathcal{T}_i}\boldsymbol{G} \big) - H\big( {\tilde{\boldsymbol{K}}_{\mathcal{T}_i}}^{(\bar{\Omega})}\boldsymbol{G}^{(\bar{\Omega})}\\
        &\qquad \qquad \qquad \qquad \quad | \; \{X_v^{(j)}:v\in\beta_{\epsilon \setminus s},\; j \in {\bar{\Omega}}\}, \\
        &\qquad \{X_v^{(j)}:v\in\beta_{\epsilon \setminus s},\; j \in {{\Omega}}\}, (\boldsymbol{M}^{(j)}:s_j \in  \beta_s) \big)\\
        & \overset{(d)}{=} H\big( \tilde{\boldsymbol{K}}_{\mathcal{T}_i}\boldsymbol{G} \big) - H\big( {\tilde{\boldsymbol{K}}_{\mathcal{T}_i}}^{(\bar{\Omega})}\boldsymbol{G}^{(\bar{\Omega})}\\
        &\qquad \qquad \qquad \qquad \quad \; |\;\{X_v^{(j)}:v\in\beta_{\epsilon \setminus s},\; j \in {\bar{\Omega}}\}\big) \\
        & \overset{(e)}{=} H\big( \tilde{\boldsymbol{K}}_{\mathcal{T}_i}\boldsymbol{G} \big) - H\big( {\tilde{\boldsymbol{K}}_{\mathcal{T}_i}}^{(\bar{\Omega})}\boldsymbol{G}^{(\bar{\Omega})}\big) \\
        & = H\big( {\tilde{\boldsymbol{K}}_{\mathcal{T}_i}}^{(\Omega)}\boldsymbol{G}^{(\Omega)} + {\tilde{\boldsymbol{K}}_{\mathcal{T}_i}}^{(\bar{\Omega})}\boldsymbol{G}^{(\bar{\Omega})} \big) - H\big( {\tilde{\boldsymbol{K}}_{\mathcal{T}_i}}^{(\bar{\Omega})}\boldsymbol{G}^{(\bar{\Omega})}\big) \\
        & \overset{(f)}{=} 0
    \end{aligned}
\end{equation}
where in the above we define the index set of observed sources as
\[
\Omega \triangleq \{j\in[|\mathcal{S}|]: s_j\in\beta_s\}, \qquad \bar{\Omega}\triangleq [|\mathcal{S}|]\setminus \Omega.
\]
Then, ${\tilde{\boldsymbol{K}}_{\mathcal{T}_i}}^{(\Omega)}$ collects the columns (equivalently, the components) $\boldsymbol{k}_{\mathcal{T}_i}^{(j)}$ of $\boldsymbol{K}_{\mathcal{T}_i}$ with $j\in\Omega$, and ${\tilde{\boldsymbol{K}}_{\mathcal{T}_i}}^{(\bar{\Omega})}$ collects those with $j\in\bar{\Omega}$. Similarly, $\boldsymbol{G}^{(\Omega)}$ consists of the rows of $\boldsymbol{G}$ indexed by $\Omega$ and $\boldsymbol{G}^{(\bar{\Omega})}$ consists of the rows indexed by $\bar{\Omega}$.
With this notation,
\[
\tilde{\boldsymbol{K}}_{\mathcal{T}_i}\boldsymbol{G}
={\tilde{\boldsymbol{K}}_{\mathcal{T}_i}}^{(\Omega)}\boldsymbol{G}^{(\Omega)}+
{\tilde{\boldsymbol{K}}_{\mathcal{T}_i}}^{(\bar{\Omega})}\boldsymbol{G}^{(\bar{\Omega})}.
\]
Step (a) follows because observing a source node $s_j \in \beta_s$ gives the eavesdropper access to all randomness generated at that source, namely $\boldsymbol{M}^{(j)}$.
Step (b) follows because knowing $\boldsymbol{M}^{(\omega)}$ for $\omega\in \Omega$, ${\tilde{\boldsymbol{K}}_{\mathcal{T}_i}}^{(\Omega)}\boldsymbol{G}^{(\Omega)} $ becomes a deterministic value.
Step (c) follows by applying step (c) of \eqref{NSNob1}, which uses the separability of transmissions associated with different source nodes.
Step (d) follows since 
$\left(\{X_v^{(j)}:v\in\beta_{\epsilon \setminus s},\; j \in {{\Omega}}\},\, (\boldsymbol{M}^{(j)}:s_j \in  \beta_s)\right)$
is independent of 
$\big({\tilde{\boldsymbol{K}}_{\mathcal{T}_i}}^{(\bar{\Omega})}\boldsymbol{G}^{(\bar{\Omega})},\;\{X_v^{(j)}:v\in\beta_{\epsilon \setminus s},\; j \in {\bar{\Omega}}\}\big)$, since the latter depend only on randomness originating from sources in $\bar{\Omega}$. 
Step (e) follows from \eqref{NSNob1}, which shows that if the adversary set contains no source nodes, then the shared key is independent of the adversary’s observation.
Finally, step (f) follows from Lemma~\ref{KG_uniform} and the fact that $\boldsymbol{G}^{(\bar{\Omega})}$ is a submatrix of the Vandermonde matrix $\boldsymbol{G}$ with distinct columns. 
Namely, we set $A$ in Lemma~\ref{KG_uniform} to be ${\tilde{\boldsymbol{K}}_{\mathcal{T}_i}}^{(\bar{\Omega})}$ and $B$ to be $\boldsymbol{G}^{(\bar{\Omega})}$. Hence,
${\tilde{\boldsymbol{K}}_{\mathcal{T}_i}}^{(\bar{\Omega})}\boldsymbol{G}^{(\bar{\Omega})}$
is uniformly distributed over $\mathbb{F}_q^{(d-\ell)\times (|\mathcal{S}|-x)}$. Since adding an independent random term only translates a uniform distribution without changing it, the sum
\[
{\tilde{\boldsymbol{K}}_{\mathcal{T}_i}}^{(\Omega)}\boldsymbol{G}^{(\Omega)}
+
{\tilde{\boldsymbol{K}}_{\mathcal{T}_i}}^{(\bar{\Omega})}\boldsymbol{G}^{(\bar{\Omega})}
\]
is also uniformly distributed over the same space. Therefore, both terms have the same entropy, and step (f) follows.
\end{proof}

\section{Multiple Key-Cast in Networks with partially-connected Nodes} \label{PCsec}
In this section, we relax the $d$-vertex connectivity assumption considered previously and study secure multiple key-cast in networks where only the terminal nodes are $d$-vertex connected from the source. Intermediate nodes that do not satisfy this connectivity requirement are referred to as partially-connected nodes, as their limited connectivity may restrict the amount of independent information they can receive from the source. We start by considering single-source networks.
We show that, despite the presence of partially-connected nodes, secure multiple key-cast remains feasible under a structural condition on the network. In particular, when no node receives input from more than $d-\ell$ partially-connected nodes, it is possible to design a coding scheme that achieves a positive key-cast rate while guaranteeing perfect secrecy against any eavesdropper observing up to $\ell<d$ nodes. 
We further generalize these results in two directions. First, we extend the analysis to the multi-source setting, where secrecy is guaranteed even when the eavesdropper may observe an arbitrary subset of source nodes, provided that the number of sources exceeds the eavesdropping capability. Second, we relax the structural condition and analyze the impact of nodes that receive inputs from more than $d-\ell$ partially-connected nodes, showing that a positive key-cast rate remains achievable with a rate degradation that depends on the extent of such violations.
In what follows, a node is referred to as partially-connected if it is less that $d$ vertex connected from any source node $s \in S$.

\begin{theorem}[Partially-connected single-source setting]
\label{T3}
    Let $\mathcal{I}=\bigl(\mathcal{V},\mathcal{E},\mathcal{B},\mathcal{S}=\{s\},{\mathcal{T}=\{\mathcal{T}_i\}_{i=1}^m}\bigr)$ be an instance of the Secure Multiple Key-Cast problem with
    {\(\mathcal{B}=
    \bigl\{
    \beta \subseteq \mathcal{V} \setminus \{s\}: |\beta|\le \ell
    \bigr\},
    \)}
    in which all terminal nodes are $d$-vertex connected from the source $s$. 
    Let $\hat{d}$ be the minimum vertex-connectivity from $s$ taken over all network nodes.
    If no $d$-vertex connected node receives input from more than $z \le d-\ell$ partially-connected nodes, then there exists a coding scheme that achieves key rate 
    $$\frac{d-\ell-z+1}{d(d - \hat{d} )+1}.$$
\end{theorem}

\begin{proof}
    Consider a network where each non-terminal node $v_i \in \mathcal{V}$ is assigned two Vandermonde vectors $\boldsymbol{v}_{i}^{(M)} \in \mathbb{F}_q^{(d-z)\times 1}$ and $\boldsymbol{v}_i^{(R)}\in \mathbb{F}_q^{d\times 1}$ indexed by $i$. 
    Furthermore, for each terminal set  $\mathcal{T}_i \subseteq \mathcal{T}$, all terminal nodes $v_t \in \mathcal{T}_i$ share the same Vandermonde vector, denoted by $\boldsymbol{v}_{\mathcal{T}_i}^{(M)} \in \mathbb{F}_q^{(d-z)\times 1}$ and $\boldsymbol{v}_{\mathcal{T}_i}^{(R)}\in \mathbb{F}_q^{d\times 1}$.
    
    Define the key for terminal set $\mathcal{T}_i$ as
    \begin{equation}\label{Theorem1_key-def}
        \boldsymbol{k}_{\mathcal{T}_i} = [\boldsymbol{M} \boldsymbol{v}_{\mathcal{T}_i}^{(M)} ]_{1:d-\ell-z+1}+[\boldsymbol{R} \boldsymbol{v}_{\mathcal{T}_i}^{(R)} ]_{1:d-\ell-z+1},
    \end{equation}
    where $[\boldsymbol{A}]_{1:d-\ell-z+1}$ denotes the first $d-\ell-z+1$ entries of $\boldsymbol{A}$, and $\boldsymbol{M} \in \mathbb{F}_q^{(d-z) \times (d-z)}$ and $\boldsymbol{R} \in \mathbb{F}_q^{d \times d}$ are independent symmetric random matrices of size $(d-z) \times (d-z)$ and $d \times d$, respectively, with each entry in their upper triangular part independently and uniformly chosen from $\mathbb{F}_q$.
    
    We start by defining our encoding scheme.
    
    \subsection*{Encoding:}

\paragraph*{Step 1}
For each node \(j \in \mathcal{V}\setminus\{s\}\), let \(0 \leq c_j \leq d\) denote the number of edges from the source node \(s\) to node \(j\). For each \(d\)-vertex-connected node \(j \in \mathcal{N}(s)\), the source uses these \(c_j\) edges to transmit two sets of \(c_j\) symbols generated using independent Vandermonde matrices
\(
\boldsymbol{V}_{s\to j}^{(M)} \in \mathbb{F}_q^{(d-z)\times c_j}
\)
and
\(
\boldsymbol{V}_{s\to j}^{(R)} \in \mathbb{F}_q^{d\times c_j},
\)
constructed from distinct Vandermonde vectors that are not used elsewhere in the network. Specifically, if $j$ is $d$-vertec connected, the source transmits
\begin{equation}\label{spurcePIb}
    \boldsymbol{\tau}_{s \to j}^{(M)}
    =
    \boldsymbol{v}_j^T \boldsymbol{M}\boldsymbol{V}_{s\to j}^{(M)},
    \qquad
    \boldsymbol{\tau}_{s \to j}^{(R)}
    =
    \boldsymbol{v}_j^T \boldsymbol{R}\boldsymbol{V}_{s\to j}^{(R)}.
\end{equation}
If node \(j\) is partially-connected, then the source sends the following vector of size $d$ to $j$:
\begin{equation}
    \boldsymbol{s}_{j}^{(R)}
    =
    \boldsymbol{R}\boldsymbol{v}_{j}^{(R)}.
\end{equation}

\paragraph*{Step 2}
    By induction, let \( j \in \mathcal{V} \) be the next node in the topological order such that every preceding node \( j' \) has either received its shares or is the source node. For each \( d \)-vertex connected node \( j' \), these shares are
    \[
    \boldsymbol{s}_{j'}^{(M)} = \boldsymbol{M}\boldsymbol{v}_{j'}^{(M)}, \quad \text{and} \quad \boldsymbol{s}_{j'}^{(R)} = \boldsymbol{R}\boldsymbol{v}_{j'}^{(R)}.
    \]
    For each partially-connected node \( j' \), we classify it as either type A or type B. Type-A partially-connected nodes are those that are either directly connected to the source or are able to recover \(\boldsymbol{R}\boldsymbol{v}_{j'}^{(R)}\); type-B partially-connected nodes are the remaining ones. The share for a type-A partially-connected node is
    \(
    \boldsymbol{s}_{j'}^{(R)} = \boldsymbol{R}\boldsymbol{v}_{j'}^{(R)}.
    \)
    Alternatively, if $j'$ is type B partially-connected node, then the share consists of the set \(\mathcal{s}_{j'}^{(R)}\), defined as all vectors \(\boldsymbol{R}\boldsymbol{v}_{j''}^{(R)}\) corresponding to \( d \)-vertex connected or type-A partially-connected nodes \( j'' \) that are either directly connected to \( j' \) or are connected to \( j' \) through one or more partially-connected nodes. Namely, define
{\footnotesize{\begin{equation}
\mathcal{D}(j') \triangleq
\left\{
\begin{array}{c@{\;}l}
j'' \in \mathcal{V}:
&
\begin{array}{l}
 j'' \text{ is either } d\text{-vertex connected from } s \text{ or} \\
 j'' \text{ is a type-A partially-connected node,} \\
 \text{and there exists a directed path from } j''  \\
 \text{to } j'\text{ whose internal nodes, if any, are}\\
 \text{all type-B partially-connected nodes}
\end{array}
\end{array}
\right\}.
\end{equation}}}
    Then,
    \begin{equation}
    \mathcal{s}_{j'}^{(R)} \triangleq \big\{ \boldsymbol{R}\boldsymbol{v}_{j''}^{(R)} \;:\; j'' \in \mathcal{D}(j') \big\}.
    \end{equation}

\paragraph*{Case 1} If $j$ and all nodes in its parent set are $d$-vertex-connected, each non-source parent node $j_p \in \mathcal{P}(j) \setminus \{s\}$ sends 
    $$
    \tau^{(M)}_{j_p \to j} = {\boldsymbol{v}_{j}^{(M)}}^T \boldsymbol{M} \boldsymbol{v}_{{j}_p}^{(M)}, 
    $$
    and
    $$
    \tau^{(R)}_{j_p \to j} = {\boldsymbol{v}_{j}^{(R)}}^T \boldsymbol{R} \boldsymbol{v}_{{j}_p}^{(R)},
    $$
    to $j$. Note that the possible transmissions from the source to node $j$ have already been specified in Step~2. Then, node $j$ stacks all received symbols from the source and its parents to reconstruct the following $1 \times d$ row vectors:
    $$
    \boldsymbol{T}^{(M)}_{j} = {\boldsymbol{v}_{j}^{(M)}}^T \boldsymbol{M} \boldsymbol{V}_{j}^{(M)},
    $$
    and
    $$
    \boldsymbol{T}^{(R)}_{j} = {\boldsymbol{v}_j^{(R)}}^T \boldsymbol{R} \boldsymbol{V}_j^{(R)},
    $$
where
\(
\boldsymbol{V}_j^{(M)} \in \mathbb{F}_q^{(d-z)\times d}
\)
and
\(
\boldsymbol{V}_j^{(R)} \in \mathbb{F}_q^{d\times d}
\)
are Vandermonde matrices whose first \(c_j\) columns are given by
\(
\boldsymbol{V}_{s\to j}^{(M)}
\)
and
\(
\boldsymbol{V}_{s\to j}^{(R)},
\)
respectively, while the remaining columns correspond to the Vandermonde vectors
\(
\{\boldsymbol{v}_{j_p}^{(M)}\}_{j_p \in \mathcal{P}(j)}
\)
and
\(
\{\boldsymbol{v}_{j_p}^{(R)}\}_{j_p \in \mathcal{P}(j)},
\)
respectively.
Since node \(j\) is \(d\)-vertex connected, all associated Vandermonde indices are distinct. Hence, \(\boldsymbol{V}_j^{(R)}\) is full rank and invertible, while every \((d-z)\times(d-z)\) submatrix of \(\boldsymbol{V}_j^{(M)}\) is invertible. Therefore, node \(j\) can uniquely recover
\begin{align}\label{Theorem1_eq_t2_case1}
    \boldsymbol{s}_{j}^{(M)}
    &=
    \left(
    [\boldsymbol{T}_{j}^{(M)}]_{1:{d-z}}
    \left[\boldsymbol{V}_{j}^{(M)}\right]_{:,1:d-z}^{-1}
    \right)^T
    \nonumber\\
    &=
    \boldsymbol{M}\boldsymbol{v}_{j}^{(M)}.
\end{align}
Similarly,
\begin{align}
    \boldsymbol{s}_{j}^{(R)}
    &=
    \left(
    \boldsymbol{T}_{j}^{(R)}
    \left(\boldsymbol{V}_{j}^{(R)}\right)^{-1}
    \right)^T
    \nonumber\\
    &=
    \left(
    {\boldsymbol{v}_{j}^{(R)}}^T
    \boldsymbol{R}
    \boldsymbol{V}_{j}^{(R)}
    \left(\boldsymbol{V}_{j}^{(R)}\right)^{-1}
    \right)^T
    \nonumber\\
    &=
    \boldsymbol{R}\boldsymbol{v}_{j}^{(R)}.
\end{align}
    
    \paragraph*{Case 2}
    Here, we consider the case where \(j\) is a partially-connected node and every non-source parent \(j_p \in \mathcal{P}(j) \setminus \{s\}\) is \(d\)-vertex connected. If \(j \in \mathcal{N}(s)\), then \(j\) is a type-A partially-connected node and has already received
    \[
    \boldsymbol{s}_j^{(R)} = \boldsymbol{R}\boldsymbol{v}_j^{(R)}
    \]
    directly from the source.
    Otherwise, if \(j \notin \mathcal{N}(s)\), node \(j\) receives the vectors
    \(\boldsymbol{R}\boldsymbol{v}_{j_p}^{(R)}\) from its parents \(j_p \in \mathcal{P}(j)\) and forms the collection
    \[
    \mathcal{s}_j^{(R)}
    =
    \big\{
    \boldsymbol{R}\boldsymbol{v}_{j_p}^{(R)}
    :
    j_p \in \mathcal{D}(j)
    \big\},
    \]
    where \(\mathcal{D}(j)=\mathcal{P}(j)\). Since \(|\mathcal{D}(j)|<d\), node \(j\) cannot recover \(\boldsymbol{s}_j^{(R)} = \boldsymbol{R}\boldsymbol{v}_j^{(R)}\) and is therefore a type-B partially-connected node.

\paragraph*{Case 3}
In this case, \(j\) is a partially-connected node and its parent set contains a subset of partially-connected nodes \(\mathcal{J}_{PI}\subseteq\mathcal{P}(j)\). Let \(\mathcal{J}_{PI}^{\text{type-A}}\subseteq \mathcal{J}_{PI}\) and \(\mathcal{J}_{PI}^{\text{type-B}}\subseteq \mathcal{J}_{PI}\) denote the subsets of type-A and type-B partially-connected parent nodes, respectively.
Define $\mathcal{D}(j)$ as the set of \( d \)-vertex connected or type-A partially-connected nodes \( j' \) that are either directly connected to \( j \) or are connected to \( j \) through one or more type-B partially-connected nodes.
If $|\mathcal{D}(j)|<d$, node $j$ is a type-B partially-connected node. It receives the vectors $\boldsymbol{R} \boldsymbol{v}_{j_p}^{(R)}$ from each $j_p \in \mathcal{P}(j) \setminus \mathcal{J}_{PI}^{\text{type-B}}$.
For nodes is $\mathcal{J}_{PI}^{\text{type-B}}$, assume an arbitrary ordering: $\mathcal{J}_{PI}^{\text{type-B}} = \{p_1, p_2, \dots, p_{|\mathcal{J}_{PI}^{\text{type-B}}|}\}$. Each type-B partially-connected parent $p_i$ transmits only the vectors from its set $\mathcal{D}(p_i)$ that are neither in the parent set of $j$ nor already covered by a preceding type-B parent. Specifically, $p_i$ transmits the set of vectors:
$$
\left\{ \boldsymbol{R}\boldsymbol{v}_{j'}^{(R)} \;:\; j' \in \mathcal{D}(p_i) \setminus \left( \mathcal{P}(j) \cup \bigcup_{r=1}^{i-1} \mathcal{D}(p_r) \right) \right\}.
$$
By taking the union of the individually received vectors from non-type-B parents and the mutually disjoint subsets of vectors from the type-B parents, node $j$ forms its complete collection without receiving duplicates:
$$
\mathcal{s}_{j}^{(R)} \triangleq \big\{ \boldsymbol{R}\boldsymbol{v}_{j'}^{(R)} \;:\; j' \in \mathcal{D}(j) \big\}.
$$
Note that while a type-B partially-connected node applies this set-difference rule to filter out duplicate transmissions to its downstream type-B partially-connected children, taking the union of these mutually disjoint transmissions perfectly reconstructs the accumulated sets. Hence, as the transmissions flow down any directed path of type-B partially-connected nodes, these distinct subsets seamlessly aggregate. By induction, the final node $j$ receives all distinct vectors originating from $d$-vertex connected or type-A partially-connected nodes that feed into any preceding type-B node on the path. Therefore, $j$ can construct $\mathcal{s}_{j}^{(R)}$ exactly according to the definition of $\mathcal{D}(j)$.

Alternatively, if \(|\mathcal{D}(j)|\ge d\) or \(j \in \mathcal{N}(s)\), node \(j\) is classified as a type-A partially-connected node. In what follows, we will show that $j$ is able to recover its exact share:
\[
\boldsymbol{s}_j^{(R)}=\boldsymbol{R}\boldsymbol{v}_j^{(R)}.
\]
If \(j\in\mathcal{N}(s)\), then \(\boldsymbol{s}_j^{(R)}\) has already been received directly from the source. Otherwise, let \(\widehat{\mathcal{D}}(j)\subseteq\mathcal{D}(j)\) be any subset of size \(d\). For each \(j_p\in\widehat{\mathcal{D}}(j)\), node \(j\) receives 
\[
    {\boldsymbol{v}_{j}^{(R)}}^T \boldsymbol{R} \boldsymbol{v}_{j_p}^{(R)} .
\]
The availability of these values is justified as follows: if \(j_p\) is directly connected to \(j\), it can compute and transmit \({\boldsymbol{v}_{j}^{(R)}}^T \boldsymbol{R} \boldsymbol{v}_{j_p}^{(R)}\) directly from its share $\boldsymbol{R}\boldsymbol{v}_{j_p}^{R}$. If \(j_p\) is connected to \(j\) via a path of type-B partially-connected nodes, those type-B nodes already store the full vector \(\boldsymbol{R}\boldsymbol{v}_{j_p}^{(R)}\) received from their ancestors. Therefore, they can locally compute and transmit \({\boldsymbol{v}_{j}^{(R)}}^T \boldsymbol{R} \boldsymbol{v}_{j_p}^{(R)}\) to node \(j\).

Let \(\widehat{\boldsymbol{V}}_j\) be the \(d \times d\) Vandermonde matrix whose columns are \(\{\boldsymbol{v}_{j_p}^{(R)}:j_p\in\widehat{\mathcal{D}}(j)\}\). By stacking the \(d\) received symbols $\{{\boldsymbol{v}_{j}^{(R)}}^T \boldsymbol{R} \boldsymbol{v}_{j_p}^{(R)}\}_{j_p \in \widehat{\mathcal{D}}(j)}$ into a row vector, node \(j\) forms \({\boldsymbol{v}_j^{(R)}}^T \boldsymbol{R} \widehat{\boldsymbol{V}}_j\). Since \(\widehat{\boldsymbol{V}}_j\) is invertible, node \(j\) can find its share as:
\begin{align*}
    \boldsymbol{s}_j^{(R)}
    &=
    \left(
    {\boldsymbol{v}_j^{(R)}}^T
    \boldsymbol{R}
    \widehat{\boldsymbol{V}}_j
    \widehat{\boldsymbol{V}}_j^{-1}
    \right)^T  \\
    &=
    \boldsymbol{R}^T \boldsymbol{v}_j^{(R)} \\
    &=
    \boldsymbol{R}\boldsymbol{v}_j^{(R)}.
\end{align*}

    \paragraph*{Case 4}
Suppose that $j$ is $d$-vertex connected and that its parent set contains a set of partially-connected nodes $\mathcal{J}_{PI}\subseteq \mathcal{P}(j)$, with $|\mathcal{J}_{PI}|\le z$. Let $\mathcal{J}_{PI}^{\text{type-A}}\subseteq \mathcal{J}_{PI}$ and $\mathcal{J}_{PI}^{\text{type-B}}\subseteq \mathcal{J}_{PI}$ denote the subsets of type-A and type-B partially-connected parent nodes, respectively. Node $j$ receives
    $$
    \tau^{(M)}_{j_p \to j} = {\boldsymbol{v}_{j}^{(M)}}^T \boldsymbol{M} \boldsymbol{v}_{{j}_p}^{(M)}, 
    $$
from all $d$-vertex-connected parents $j_p\in \mathcal{P}(j)\setminus (\mathcal{J}_{PI}\cup\{s\})$, together with the $c_j$ symbols received directly from the source in \eqref{spurcePIb}. Then, node $j$ can stack these symbols as
$$
\boldsymbol{T}_{j}^{(M)}
=
{\boldsymbol{v}_{j}^{(M)}}^T
\boldsymbol{M}
\boldsymbol{V}_{j}^{(M)},
$$
where $\boldsymbol{V}_{j}^{(M)} \in \mathbb{F}_q^{(d-z)\times (d-z)}$ is a Vandermonde matrix whose first $c_j$ columns are given by $\boldsymbol{V}_{s\to j}^{(M)}$, while the remaining $d-c_j-z$ columns correspond to the Vandermonde vectors $\{\boldsymbol{v}_{j_p}^{(M)}\}_{j_p \in \mathcal{P}(j)\setminus (\mathcal{J}_{PI}\cup\{s\})}$. Since the columns of $\boldsymbol{V}_{j}^{(M)}$ have distinct Vandermonde indices, $\boldsymbol{V}_{j}^{(M)}$ is invertible. Hence node $j$ recovers
\begin{align}
    \boldsymbol{s}_{j}^{(M)}
    &=
    \left(
    \boldsymbol{T}_{j}^{(M)}
    \left(\boldsymbol{V}_{j}^{(M)}\right)^{-1}
    \right)^T
    \nonumber\\
    &=
    \left(
    {\boldsymbol{v}_{j}^{(M)}}^T
    \boldsymbol{M}
    \boldsymbol{V}_{j}^{(M)}
    \left(\boldsymbol{V}_{j}^{(M)}\right)^{-1}
    \right)^T
    \nonumber\\
    &=
    \boldsymbol{M}\boldsymbol{v}_{j}^{(M)} .
\end{align}

From each $d$-vertex connected or type-A partially-connected parent node $j_p \in \mathcal{P}(j) \setminus (\mathcal{J}_{PI}^{\text{type-B}}\cup\{s\})$, node $j$ receives
    $$
    \tau^{(R)}_{j_p \to j} = {\boldsymbol{v}_{j}^{(R)}}^T \boldsymbol{R} \boldsymbol{v}_{{j}_p}^{(R)}.
    $$
Without loss of generality, we
restrict our attention to a minimal subset of these parents, denoted by
\(\widehat{\mathcal{J}}_{PI}^{\mathrm{type-B}}
\subseteq \mathcal{J}_{PI}^{\mathrm{type-B}}\), such that removing the remaining
nodes in
\(\mathcal{J}_{PI}^{\mathrm{type-B}}
\setminus \widehat{\mathcal{J}}_{PI}^{\mathrm{type-B}}\)
from the network does not reduce the maximum number of vertex-disjoint paths
from the source \(s\) to node \(j\).

\begin{claim}\label{claim:eta-selection}
Let \(j\) be a \(d\)-vertex-connected node from the source \(s\), and let \(\mathcal{J}_{PI}^{\mathrm{type-B}}\subseteq\mathcal P(j)\) denote the set of its type-B partially-connected parents. 
Let \(\widehat{\mathcal{J}}_{PI}^{\mathrm{type-B}}\subseteq\mathcal{J}_{PI}^{\mathrm{type-B}}\) be a minimal subset such that the removal of rest in $\mathcal{J}_{PI}^{\mathrm{type-B}} \setminus \widehat{\mathcal{J}}_{PI}^{\mathrm{type-B}}$ does not reduce the maximum number of vertex-disjoint paths from \(s\) to \(j\).
Then, for every \(j_{PI}\in\widehat{\mathcal{J}}_{PI}^{\mathrm{type-B}}\), there exists an index \(\eta(j_{PI})\) such that
\[
    \boldsymbol R\boldsymbol v_{\eta(j_{PI})}^{(R)}
    \in \mathcal s_{j_{PI}}^{(R)},
    \qquad
    \eta(j_{PI})\notin \mathcal P(j),
\]
and the indices \(\{\eta(j_{PI}):j_{PI}\in
\widehat{\mathcal{J}}_{PI}^{\mathrm{type-B}}\}\) are mutually distinct.
\end{claim}

\begin{proof}[Proof of Claim \ref{claim:eta-selection}]
By the definition of \(\widehat{\mathcal{J}}_{PI}^{\mathrm{type-B}}\), this subset is minimal with respect to preserving the \(d\) vertex-disjoint paths from \(s\) to \(j\). Therefore, each parent \(j_{PI} \in \widehat{\mathcal{J}}_{PI}^{\mathrm{type-B}}\) must lie on a strictly distinct vertex-disjoint path to \(j\); otherwise, any redundant node could be removed without reducing the total number of disjoint paths, contradicting minimality.

Fix any \(j_{PI}\in\widehat{\mathcal{J}}_{PI}^{\mathrm{type-B}}\). Since \(j\) is
\(d\)-vertex connected from \(s\), there exist \(d\) vertex-disjoint paths
from \(s\) to \(j\). 
Remove \(d-1\) such paths that do not pass through
\(j_{PI}\), together with all vertices on those paths. The remaining path
from \(s\) to \(j\) must pass through \(j_{PI}\). Let
\(\eta(j_{PI})\) denote the last \(d\)-vertex-connected or type-A partially-connected node on this remaining path before the path enters the consecutive
sequence of partially-connected nodes leading to \(j_{PI}\). By
construction,
\[
\eta(j_{PI})\notin \mathcal{P}(j).
\]
Moreover, by the forwarding rule for type-B partially-connected nodes, whenever
there is an edge from a type-B partially-connected node \(a\) to another
type-B partially-connected node \(b\), all \(\mathcal{s}^{(R)}\)-type information
available at \(a\) is forwarded to \(b\), i.e.,
\[
\mathcal{s}_{a}^{(R)} \subseteq \mathcal{s}_{b}^{(R)} .
\]
Hence, along the subpath
\[
\eta(j_{PI}) \rightarrow u_1 \rightarrow u_2 \rightarrow \cdots
\rightarrow u_L \rightarrow j_{PI},
\]
where \(u_1,\ldots,u_L\) are all type-B partially-connected nodes, the share
\(\boldsymbol{R}\boldsymbol{v}_{\eta(j_{PI})}^{(R)}\) is propagated
forward. Therefore,
\[
\boldsymbol{R}\boldsymbol{v}_{\eta(j_{PI})}^{(R)}
\in \mathcal{s}_{j_{PI}}^{(R)} .
\]
Since \(j_{PI}\in\widehat{\mathcal{J}}_{PI}^{\mathrm{type-B}}\) was arbitrary, such an index
\(\eta(j_{PI})\) exists for every partially-connected parent in \(\widehat{\mathcal{J}}_{PI}^{\mathrm{type-B}}\). The indices \(\eta(j_{PI})\) are all distinct since each \(\eta(j_{PI})\) is on the path from $s$ to \(j_{PI}\), which are chosen to be strictly vertex disjoint.
\end{proof}

By Claim~\ref{claim:eta-selection}, node $j$ receives from each parent $j_{PI} \in \widehat{\mathcal{J}}_{PI}^{\mathrm{type-B}}$ a single symbol projected from its accumulated set:
$$
\tau^{(R)}_{j_{PI} \to j} = {\boldsymbol{v}_{j}^{(R)}}^T \boldsymbol{R}\boldsymbol{v}_{\eta(j_{PI})}^{(R)},
$$
where $\boldsymbol{R}\boldsymbol{v}_{\eta(j_{PI})}^{(R)} \in \mathcal{s}_{j_{PI}}^{(R)}$, $\eta(j_{PI}) \notin \mathcal{P}(j)$, and $\eta(j_{PI}) \neq \eta(j'_{PI})$ for all $j_{PI} \neq j'_{PI} \in \widehat{\mathcal{J}}_{PI}^{\text{type-B}}$.

Node $j$ then stacks all received $\boldsymbol{R}$-related symbols to form
$$
\boldsymbol{T}_{j}^{(R)}
=
{\boldsymbol{v}_{j}^{(R)}}^T
\boldsymbol{R}
\widehat{\boldsymbol{V}}_{j}^{(R)},
$$
where $\widehat{\boldsymbol{V}}_{j}^{(R)}$ is defined as the block matrix:
$$
\widehat{\boldsymbol{V}}_{j}^{(R)} \triangleq \big[\, \boldsymbol{V}_{s\to j}^{(R)} \;,\; \boldsymbol{V}_{j}^{(R)} \;,\; \boldsymbol{V}_{\eta}^{(R)} \,\big],
$$
with $\boldsymbol{V}_{j}^{(R)}$ containing the column vectors $\boldsymbol{v}_{j_p}^{(R)}$ for $j_p \in \mathcal{P}(j)\setminus ({\mathcal{J}}_{PI}^{\text{type-B}}\cup \{s\})$, and $\boldsymbol{V}_{\eta}^{(R)}$ containing the column vectors $\boldsymbol{v}_{\eta(j_{PI})}^{(R)}$ for all $j_{PI} \in \widehat{\mathcal{J}}_{PI}^{\text{type-B}}$. Since all involved Vandermonde indices are mutually distinct, $\widehat{\boldsymbol{V}}_{j}^{(R)}$ is full rank and invertible. Thus, node $j$ perfectly recovers
\begin{align}
    \boldsymbol{s}_{j}^{(R)}
    &=
    \left(
    \boldsymbol{T}_{j}^{(R)}
    \left(\widehat{\boldsymbol{V}}_{j}^{(R)}\right)^{-1}
    \right)^T
    \nonumber\\
    &=
    \left(
    {\boldsymbol{v}_{j}^{(R)}}^T
    \boldsymbol{R}
    \widehat{\boldsymbol{V}}_{j}^{(R)}
    \left(\widehat{\boldsymbol{V}}_{j}^{(R)}\right)^{-1}
    \right)^T
    \nonumber\\
    &=
    \boldsymbol{R}\boldsymbol{v}_{j}^{(R)} .
\end{align}

\paragraph*{Step 3}
Let \(t\in\mathcal{T}_i\) be a terminal node. By definition, all terminal nodes are \(d\)-vertex connected from the source \(s\). Therefore, the share recovery procedure at node \(t\) is identical to that of the \(d\)-vertex connected intermediate nodes detailed in Step~2 (specifically, Cases~1 and~4). Therefore, node \(t\) successfully recovers its designated shares as:
\[
\boldsymbol{s}_{t}^{(M)} = \boldsymbol{M}\boldsymbol{v}_{\mathcal{T}_i}^{(M)}, \qquad \boldsymbol{s}_{t}^{(R)} = \boldsymbol{R}\boldsymbol{v}_{\mathcal{T}_i}^{(R)}.
\]

\subsection*{Decoding:}
    \paragraph*{Step 5} Using the recovered vectors $\boldsymbol{s}_t^{(M)}$ and $\boldsymbol{s}_t^{(R)}$, the terminal node $t \in \mathcal{T}_i$ recovers the key as  
    \begin{equation} \label{Theorem1_eq_ket}
    \boldsymbol{k}_{\mathcal{T}_i} = [(\boldsymbol{M}) \boldsymbol{v}_{\mathcal{T}_i}^{(M)} ]_{1:d-\ell-z+1}+[(\boldsymbol{R}) \boldsymbol{v}_{\mathcal{T}_i}^{(R)} ]_{1:d-\ell-z+1}
    \end{equation}

\subsection*{Key Rate}
Each terminal node \(t\in\mathcal{T}_i\subseteq\mathcal{T}\) can recover the
corresponding key defined in \eqref{Theorem1_eq_ket}. The proposed scheme
delivers \(d-\ell-z+1\) key symbols. This requires one transmission over each
network edge to deliver \(\boldsymbol{M}\boldsymbol{v}_{\mathcal{T}_i}^{(M)}\)
to the terminals.
In order to bound the number of additional transmissions needed for the
\(\boldsymbol{R}\)-related information, we consider the worst-case scenario.
Applying the process specified in step 2, every $d$-vertex connected or type-A partially-connected node $j$ is able to recover $\boldsymbol{s}_{j}^{(R)}$ by receiving one symbol from each of its parents. However, when $j$ is type-B partially-connected and $D(j)=d-1$, $j$ requires to receive $d-1$ vectors of size $d$, which is the worst case scenario. 
Using the same logic established earlier, we restrict our attention to a minimal subset of parent nodes, denoted by $\widehat{\mathcal{P}}(j) \subseteq \mathcal{P}(j)$, such that removing the remaining nodes in $\mathcal{P}(j) \setminus \widehat{\mathcal{P}}(j)$ from the network does not reduce the maximum number of vertex-disjoint paths from the source $s$ to node $j$. 
By the minimality of $\widehat{\mathcal{P}}(j)$, and using the same logic as in the proof of Claim \ref{claim:eta-selection}, one can show that each parent in this subset holds at least one unique vector of the form $\boldsymbol{R} \boldsymbol{v}_i^{(R)}$ for $i \in \mathcal{D}(j)$. Therefore, when $|\mathcal{D}(j)|=d-1$, node $j$ receives at least one unique vector of size $d$ from each of the $\hat{d}$ parents in this minimal subset.
Hence, at most $d - \hat{d} - 1$ vectors are left to be transmitted from, in worst case, single type-B parent node. 
Consequently, each edge in the network is used at most $d(d-\hat{d})$ times for the \(\boldsymbol{R}\)-related information. 

Finally, the achieved key rate is
\[
    R
    \ge
    \frac{d-\ell-z+1}{d(d-\hat d)+1}.
\]

\subsection*{Security Proof:}
    Each $d$-vertex-connected node $j \in \mathcal{V}$, including the terminal nodes, has access to:
    \begin{align}
        &(I): \boldsymbol{M}  \boldsymbol{v}_j^{(M)}, \\
        &(II): \boldsymbol{R}  \boldsymbol{v}_j^{(R)}.
    \end{align}
    Similarly, each partially-connected node \(j \in \mathcal{V}\) holds
    \begin{align}
        &(I)\;: \mathcal{s}_{j}^{(R)}
        =
        f_j(\boldsymbol{R}),
    \end{align}
    where \(f_j(\cdot)\) denotes the collection of vectors derived from
    \(\boldsymbol{R}\) available at node \(j\).
    Note that partially-connected nodes holds no information about $\boldsymbol{M}$.

    Consider any $\beta \in \mathcal{B}$ such that $\beta \cap \mathcal{T}_i =\phi$. We aim to prove that
    \begin{equation}
        I\left(\boldsymbol{k}_{\mathcal{T}_i}; \{X_v:v\in\beta\}\right) =0,  \quad \forall \mathcal{T}_i \subseteq\mathcal{T}.
    \end{equation}

\subsection*{Case 1: Eavesdropper Observes at Least One partially-connected Node}
    Let us assume an eavesdropper observes a set of nodes that includes one or more partially-connected nodes. Since such nodes may possess a significant amount of information about $\boldsymbol{R}$, we consider the worst-case scenario in which the eavesdropper has full knowledge of $\boldsymbol{R}$. Let $d'=d-z$ and $\ell' = \ell-1$. Then,  
    \begin{equation}
        \label{eq:secrecy_with_R_known}
        \begin{aligned}
        I(\boldsymbol{k}_{\mathcal{T}_i};& \{X_v:v\in\beta\})
        \le I(\boldsymbol{k}_{\mathcal{T}_i}; \{X_v:v\in\beta\}, \boldsymbol{R}) \\
        &= I\!\left(\boldsymbol{k}_{\mathcal{T}_i};
        \boldsymbol{M}\boldsymbol{V}_{\epsilon}^{(M)},\;
        \boldsymbol{R}\right) \\
        &\overset{(a)}{=} I\!\big([\boldsymbol{M}\boldsymbol{v}_{\mathcal{T}_i}^{(M)}]_{1:d'-\ell'} + [\boldsymbol{R}\boldsymbol{v}_{\mathcal{T}_i}^{(R)}]_{1:d'-\ell'};   \boldsymbol{M}\boldsymbol{V}_\epsilon^{(M)},\;
        \boldsymbol{R}\big) \\
        &\overset{(b)}{=} I\!\big([\boldsymbol{M}\boldsymbol{v}_{\mathcal{T}_i}^{(M)}]_{1:d'-\ell'} + [\boldsymbol{R}\boldsymbol{v}_{\mathcal{T}_i}^{(R)}]_{1:d'-\ell'}; \boldsymbol{M}\boldsymbol{V}_\epsilon^{(M)}\big)\\
        &+ I\!\big([\boldsymbol{M}\boldsymbol{v}_{\mathcal{T}_i}^{(M)}]_{1:d'-\ell'} + [\boldsymbol{R}\boldsymbol{v}_{\mathcal{T}_i}^{(R)}]_{1:d'-\ell'};  \boldsymbol{M}\boldsymbol{V}_\epsilon^{(M)}|\;
        \boldsymbol{R}\big)\\
        &\overset{(c)}{=} I\!\big([\boldsymbol{M}\boldsymbol{v}_{\mathcal{T}_i}^{(M)}]_{1:d'-\ell'} + [\boldsymbol{R}\boldsymbol{v}_{\mathcal{T}_i}^{(R)}]_{1:d'-\ell'};\boldsymbol{M}\boldsymbol{V}_\epsilon^{(M)}|\;
        \boldsymbol{R}\big) \\
        &\overset{(d)}{=} I\!\big([\boldsymbol{M}\boldsymbol{v}_{\mathcal{T}_i}^{(M)}]_{1:d'-\ell'}; \;   \boldsymbol{M}\boldsymbol{V}_\epsilon^{(M)}\;\big) \\
        &\overset{(e)}{=} 0\; ,
        \end{aligned}
    \end{equation}
    where $\boldsymbol{V}_\epsilon^{(M)} \in \mathbb{F}_q^{d' \times \ell'}$ is the Vandermonde matrix formed by the eavesdropped nodes, whose $m$-th column corresponds to the Vandermonde vector assigned to the $m$-th observed node. Note that there should exist at least one partially-information node, which posses no information about $\boldsymbol{M},$ in the eavesdropped set $\beta$, hence, $\boldsymbol{V}_\epsilon^{(M)}$ may have at most $\ell' = \ell-1$ linearly independent columns.
    Step (a) follows from the definition of $\boldsymbol{k}_{\mathcal{T}_i}$ in \eqref{Theorem1_key-def}.
    Step (b) follows from the chain rule.
    Step (c) follows since, conditioned on \(\boldsymbol{R}\),
    \([\boldsymbol{R}\boldsymbol{v}_{\mathcal{T}_i}^{(R)}]_{1:d'-\ell'}\) is fixed, while
    \([\boldsymbol{M}\boldsymbol{v}_{\mathcal{T}_i}^{(M)}]_{1:d'-\ell'}\) remains uniform and independent of \(\boldsymbol{R}\). Hence, their sum is independent of \(\boldsymbol{R}\), and 
    $$I\!\big([\boldsymbol{M}\boldsymbol{v}_{\mathcal{T}_i}^{(M)}]_{1:d'-\ell'} + [\boldsymbol{R}\boldsymbol{v}_{\mathcal{T}_i}^{(R)}]_{1:d'-\ell'}; \boldsymbol{M}\boldsymbol{V}_\epsilon^{(M)}\big)=0.$$
    Step (d) follows because conditioning on $\boldsymbol{R}$ makes $[\boldsymbol{R}\boldsymbol{v}_{\mathcal{T}_i}]_{1:d-\ell}$ deterministic. Furthermore, because $\boldsymbol{M}$ and $\boldsymbol{R}$ are independent, conditioning on $\boldsymbol{R}$ provides no information about $\boldsymbol{M}$, allowing the condition itself to be dropped.
    Finally, step (e) follows from Lemma \ref{main_ind}, which states that $[\boldsymbol{M} \boldsymbol{v}_{\mathcal{T}_i} ]_{1:d'-\ell'}$ is uniformly distributed over $\mathbb{F}_q^{d'-\ell'}$ given $\boldsymbol{M}\boldsymbol{V}_{\epsilon}^{(M)}$, and therefore independent of $\boldsymbol{M}\boldsymbol{V}_\epsilon^{(M)}$.

\subsection*{Case 2: Eavesdropper Observes Only Fully Connected Nodes}
    In this case, the eavesdropper’s observation is restricted to fully connected nodes, none of which correspond to partially-connected nodes.
    Since such fully connected nodes may collectively be able to recover $\boldsymbol{M}$, we consider the worst-case scenario in which the eavesdropper has full knowledge of $\boldsymbol{M}$. Then,  
    \begin{equation}
        \begin{aligned}
            I(\boldsymbol{k}_{\mathcal{T}_i}; &\{X_v:v\in\beta\}) 
            \le I(\boldsymbol{k}_{\mathcal{T}_i}; \{X_v:v\in\beta\}, \boldsymbol{M}\boldsymbol) \\
            &= I\!\left(\boldsymbol{k}_{\mathcal{T}_i}; \boldsymbol{R}\boldsymbol{V}_{\epsilon }^{(R)},\; \boldsymbol{M}\boldsymbol \; \right) \\
            &\overset{(a)}{=} I\!\big([\boldsymbol{M} \boldsymbol{v}_{\mathcal{T}_i}^{(M)}]_{1:d-\ell-z+1} + [\boldsymbol{R} \boldsymbol{v}_{\mathcal{T}_i}^{(R)} ]_{1:d-\ell-z+1}; \\
            & \qquad \qquad \qquad \qquad \qquad \qquad \qquad \; \boldsymbol{R}\boldsymbol{V}_{\epsilon }^{(R)},  \boldsymbol{M}\big) \\ 
            &\overset{(b)}{=} I\!\left([\boldsymbol{R} \boldsymbol{v}_{\mathcal{T}_i}^{(R)}]_{1:d-\ell-z+1}; \boldsymbol{R}\boldsymbol{V}_{\epsilon }^{(R)}\right) \\  
            &\overset{(c)}{=} 0,
        \end{aligned}
    \end{equation}
    where $\boldsymbol{V}_\epsilon^{(R)} \in \mathbb{F}_q^{d \times \ell}$ is the Vandermonde matrix formed by the eavesdropped nodes, whose $m$-th column corresponds to the Vandermonde vector assigned to the $m$-th observed node.
    Step (a) follows from the definition of $\boldsymbol{k}_{\mathcal{T}_i}$ in \eqref{Theorem1_key-def}.
    Step (b) follows because conditioning on $\boldsymbol{M}$ makes $[\boldsymbol{M}\boldsymbol{v}_{\mathcal{T}_i}^{(M)}]_{1:d-\ell}$ deterministic. Furthermore, because $\boldsymbol{M}$ and $\boldsymbol{R}$ are independent, conditioning on $\boldsymbol{M}$ provides no information about $\boldsymbol{R}$, allowing the condition itself to be dropped.
    Finally, step \((c)\) follows from Lemma~\ref{main_ind}, which shows that \( [\boldsymbol{R}\boldsymbol{v}_{\mathcal{T}_i}^{(R)}]_{1:d-\ell} \) is independent of \( \boldsymbol{R}\boldsymbol{V}_{\epsilon}^{(R)} \). Since the index set \(1:d-\ell-z+1\) is a subset of \(1:d-\ell\), the subvector \( [\boldsymbol{R}\boldsymbol{v}_{\mathcal{T}_i}^{(R)}]_{1:d-\ell-z+1} \) also remains independent of \( \boldsymbol{R}\boldsymbol{V}_{\epsilon}^{(R)} \), and therefore the corresponding mutual information is zero.
    
    Combining the results from Case 1 and Case 2, we have shown that regardless of which set of nodes the eavesdropper observes, whether it includes at least one partially-connected node or only fully connected nodes, the eavesdropper’s observation does not reveal any information about the key. Formally, for any possible eavesdropper observation $\beta \in \mathcal{B}$ such that $\beta \cap \mathcal{T}_i =\phi$, we have
    \[
    I\left(\boldsymbol{k}_{\mathcal{T}_i}; \{X_v:v\in\beta\}\right) = 0.
    \]
    This completes the proof of security.

\end{proof}

{\begin{remark}[Choice of the design parameter $d$ and of the subgraph $\mathcal{G}'$]\label{R1}
Note that the parameter $d$ is a design choice, whereas the connectivity parameters below are determined by the network. Moreover, by Remark~\ref{R0}, the schemes of Theorem~\ref{T1} and Theorem~\ref{T3} need not be run on the full network $\mathcal{G}=(\mathcal{V},\mathcal{E})$. For any subgraph $\mathcal{G}'=(\mathcal{V}',\mathcal{E}')\subseteq\mathcal{G}$ that contains the source $s$, all terminal nodes, and any chosen subset of the remaining (intermediate) nodes: Theorem~\ref{T1} applies directly to $\mathcal{G}'$ whenever every node of $\mathcal{G}'$ is $d$-vertex connected from $s$ within $\mathcal{G}'$, and Theorem~\ref{T3} applies whenever only the terminals of $\mathcal{G}'$ are guaranteed to be $d$-vertex connected within $\mathcal{G}'$, with the structural condition of Theorem~\ref{T3} holding for the remaining, possibly partially-connected, nodes. We therefore phrase the discussion below for a general such $\mathcal{G}'$, with the understanding that $\mathcal{G}'=\mathcal{G}$ recovers the setting considered so far.

For a fixed choice of $\mathcal{G}'$, define the following quantities, all evaluated within $\mathcal{G}'$:
\begin{itemize}
    \item $\hat d_{\mathcal{G}'}$, the minimum vertex-connectivity from $s$ taken over all nodes of $\mathcal{G}'$;
    \item $\Omega_{\mathcal{G}'}(d)$, the set of partially-connected nodes of $\mathcal{G}'$:
    $$\{v\in\mathcal{V}':\ v\ \text{is less than}\ d\text{-vertex connected from}\ s\ \text{in}\ \mathcal{G}'\}.$$
    \item $z_{\mathcal{G}'}(d)$, the largest number of partially-connected parents (i.e., parents in $\Omega_{\mathcal{G}'}(d)$) of any $d$-vertex connected node of $\mathcal{G}'$;
    \item $d_{\mathcal{T}}(\mathcal{G}')$, the minimum vertex-connectivity from $s$ within $\mathcal{G}'$ taken over the terminal nodes.
\end{itemize}
Increasing $d$ enlarges the set $\Omega_{\mathcal{G}'}(d)$, and thus, as $d$ grows, a $d$-vertex connected node of $\mathcal{G}'$ may acquire additional partially-connected parents, i.e., $z_{\mathcal{G}'}(d)$ may grow with $d$.
Since all terminals are required to be $d$-vertex connected from $s$ within $\mathcal{G}'$, and since $\ell<d$ is required throughout, the admissible values of $d$ satisfy $\ell+\max\{z_{\mathcal{G}'}(d),1\}\le d\le d_{\mathcal{T}}(\mathcal{G}')$.
The lower bound here combines the structural condition $z_{\mathcal{G}'}(d)\le d-\ell$ of Theorem~\ref{T3} (applied to $\mathcal{G}'$) with the requirement $\ell<d$, and guarantees a positive numerator in the rate expression below.
For fixed $\mathcal{G}'$, the rate of Theorem~\ref{T3} is therefore to be read as a function of $d$, namely
\[
R_{\mathcal{G}'}(d)\triangleq\frac{d-\ell-z_{\mathcal{G}'}(d)+1}{d\bigl(d-\hat d_{\mathcal{G}'}\bigr)+1},
\]
and, by Remark~\ref{R0}, the best rate guaranteed by our schemes is obtained by jointly optimizing $R_{\mathcal{G}'}(d)$ over the admissible range of $d$ \emph{and} over the choice of subgraph $\mathcal{G}'$ itself, i.e., by evaluating
\[
R^{\star}\triangleq\max_{\mathcal{G}'\in\mathcal{A}}\ \max_{d}\ R_{\mathcal{G}'}(d),
\]
where $\mathcal{A}$ denotes the family of \emph{admissible} subgraphs, i.e., those $\mathcal{G}'\subseteq\mathcal{G}$ that contain the source $s$ (or all sources, in the multi-source case) and all terminal nodes, and for which some choice of $d$ makes Theorem~\ref{T1} or Theorem~\ref{T3} applicable to $\mathcal{G}'$, as described above. For each $\mathcal{G}'\in\mathcal{A}$, the inner maximization over $d$ is subject to $\ell+\max\{z_{\mathcal{G}'}(d),1\}\le d\le d_{\mathcal{T}}(\mathcal{G}')$.
Two regimes arise, according to whether $\mathcal{A}$ contains a subgraph whose minimum vertex-connectivity exceeds $\ell$.

\emph{(i) There exists $\mathcal{G}'\in\mathcal{A}$ with $\hat d_{\mathcal{G}'}>\ell$.} We first observe that partially-connected nodes come at a cost in rate. For any $\mathcal{G}'\in\mathcal{A}$ and any $d>\hat d_{\mathcal{G}'}$, if 
$z_{\mathcal{G}'}(d)\ge1$, then Theorem~\ref{T3} gives
\[
R_{\mathcal{G}'}(d)\;=\; \frac{d-\ell-z_{\mathcal{G}'}(d)+1}{d\bigl(d-\hat d_{\mathcal{G}'}\bigr)+1}\;\le\;\frac{d-\ell}{d+1}\;<\;1 ,
\]
where the inequality uses $z_{\mathcal{G}'}(d)\ge1$ together with $d(d-\hat d_{\mathcal{G}'})+1\ge d+1$. Thus, whenever $z_{\mathcal{G}'}(d)\ge1$ the resulting rate is less than $1$.

It is therefore preferable to avoid partially-connected nodes altogether, which for a given $\mathcal{G}'$ means restricting to $d\le\hat d_{\mathcal{G}'}$, so that Theorem~\ref{T1} applies and yields the rate $d-\ell$. Since this rate increases with $d$, the best choice for a given $\mathcal{G}'$ is the largest $d$ for which Theorem~\ref{T1} still applies, namely $d=\hat d_{\mathcal{G}'}$, giving rate $\hat d_{\mathcal{G}'}-\ell$. This value is at least $1$ whenever $\hat d_{\mathcal{G}'}>\ell$, and hence is no smaller than the rate of less than $1$ obtained for any $d>\hat d_{\mathcal{G}'}$ and $z_{\mathcal{G}'}(d)\ge1$.

Optimizing over $\mathcal{G}'$ as well, the best rate is thus attained by the subgraph with the largest such $\hat d_{\mathcal{G}'}$: setting
\[
\mathcal{G}^{\star}\in\arg\max_{\mathcal{G}'\in\mathcal{A}:\ \hat d_{\mathcal{G}'}>\ell}\hat d_{\mathcal{G}'}, \qquad d^{\star}=\hat d_{\mathcal{G}^{\star}},
\]
achieves $R^{\star}=\hat d_{\mathcal{G}^{\star}}-\ell$.
By Lemma~\ref{L:converse}, this rate cannot be improved in general: there exist $d$-vertex connected instances whose secure key-cast capacity is at most $d-\ell$, so the bound $\hat d_{\mathcal{G}^{\star}}-\ell$ established above is the best possible.

\emph{(ii) Every $\mathcal{G}'\in\mathcal{A}$ satisfies $\hat d_{\mathcal{G}'}\le\ell$.} Here, for every $\mathcal{G}'\in\mathcal{A}$, every admissible $d$ satisfies $d>\ell\ge\hat d_{\mathcal{G}'}$, so the $\hat d_{\mathcal{G}'}$-vertex connected nodes of $\mathcal{G}'$ are necessarily partially-connected, and Theorem~\ref{T1} is not applicable on any admissible subgraph for any choice of $d$.
The best rate guaranteed by Theorem~\ref{T3} is then the value of the joint discrete optimization problem
\[
\begin{aligned}
R^{\star}=\max_{\mathcal{G}'\in\mathcal{A}}\max_{d}\quad & \frac{d-\ell-z_{\mathcal{G}'}(d)+1}{d(d-\hat d_{\mathcal{G}'})+1}\\[0.2em]
\text{s.t.}\quad & \ell+\max\{z_{\mathcal{G}'}(d),1\}\;\le\; d\;\le\; d_{\mathcal{T}}(\mathcal{G}'),
\end{aligned}
\]
where the feasibility constraint is exactly the structural condition of Theorem~\ref{T3} (applied to $\mathcal{G}'$).
Here the rate is not monotone in $d$, since increasing $d$ grows the denominator while the numerator is governed by $d-z_{\mathcal{G}'}(d)$, and neither the smallest nor the largest admissible $d$ need be optimal.
For instance, for a subgraph $\mathcal{G}'$ with $\hat d_{\mathcal{G}'}=1$, $\ell=2$, and $z_{\mathcal{G}'}(d)=1$ for $d\in\{3,4,5\}$, the corresponding rates are $1/7$, $2/13$, and $3/21$, so that the smallest admissible choice $d=3$ is not optimal and $d=4$ is preferable.
Since the choice of $\mathcal{G}'$ likewise affects $\hat d_{\mathcal{G}'}$, $z_{\mathcal{G}'}(d)$, and $d_{\mathcal{T}}(\mathcal{G}')$, the two maximizations must be carried out jointly. We leave a full characterization of this joint optimization, including efficient algorithms for identifying an optimal subgraph $\mathcal{G}'$, to future work.
\end{remark}}

\begin{corollary}[{Partially-connected multiple-source setting}]\label{C3}
Let $\mathcal{I}=(\mathcal{V},\mathcal{E},\mathcal{B},\mathcal{S},{\mathcal{T}=\{\mathcal{T}_i\}_{i=1}^m})$ be an instance of the secure multiple key-cast problem with a set of sources $\mathcal{S}=\{s_1,\dots,s_{|\mathcal{S}|}\}$. For any integer $x <|\mathcal{S}|$, let
\(
\mathcal{B}=
\bigl\{
\beta \subseteq \mathcal{V}: |\beta\setminus \mathcal{S}|\le \ell,\ |\beta \cap \mathcal{S}| \le x
\bigr\}.
\)
Assume that all terminal nodes are $d$-vertex connected from each source $s \in \mathcal{S}$. Let $\hat{d}$ be the minimum vertex-connectivity from any $s \in \mathcal{S}$ over all network nodes $\mathcal{V}\setminus \mathcal{S}$. If no $d$-vertex connected node receives input from more than $z \le d-\ell$ partially-connected nodes, then there exists a coding scheme that achieves key rate
\[
R = \frac{(d-\ell-z+1)(|\mathcal{S}|-x)}{(d(d - \hat{d} )+1)|\mathcal{S}|}.
\]
In particular, $x$ denotes the maximum number of source nodes that can be included in the eavesdropped set, while $x=0$ prohibits eavesdropping on any source node.
\end{corollary}

\begin{proof}
Consider a network in which not all nodes are necessarily $d$-vertex connected from every source $s_u \in \mathcal{S}$, but the partially-connected condition of Theorem~\ref{T3} is satisfied. That is, for each source $s\in \mathcal{S}$, no node receives inputs from more than $z \le d-\ell$ partially-connected nodes. Under this condition, the assumptions of Theorem~\ref{T3} hold for each source individually, and thus every source $s_j$ can independently and securely disseminate its key $\boldsymbol{k}_{\mathcal{T}_i}^{(j)}$ to all terminals $t \in \mathcal{T}_i$.

Define the $(d-\ell) \times |\mathcal{S}|$ matrix
\[
\tilde{\boldsymbol{K}}_{\mathcal{T}_i}
\triangleq
\big[ \boldsymbol{k}_{\mathcal{T}_i}^{(1)}, \dots , \boldsymbol{k}_{\mathcal{T}_i}^{(|\mathcal{S}|)} \big].
\]
The key for terminal set $\mathcal{T}_i$ is then defined as
\begin{equation}\label{keyCor1}
\boldsymbol{k}_{\mathcal{T}_i}
= \tilde{\boldsymbol{K}}_{\mathcal{T}_i}\boldsymbol{G}.
\end{equation}
where $\boldsymbol{G}$ is a fixed $|\mathcal{S}| \times (|\mathcal{S}|-x)$ Vandermonde matrix.

Using the same encoding and decoding procedures established in the proof of Theorem \ref{T1}, each terminal recovers all values $\{\boldsymbol{k}_{\mathcal{T}_i}^{(j)} : s_j \in \mathcal{S}\}$ and constructs its final key as \eqref{keyCor1}.

\subsection*{Key-Rate:}
Each terminal $t \in \mathcal{T}_i$ combines the keys generated by all $|\mathcal{S}|$ source nodes. Because the eavesdropper can observe at most $x$ source nodes, at least $|\mathcal{S}|-x$ of these source-generated keys remain secure, with each key containing $d-\ell-z+1$ symbols. Consequently, by combining the $|\mathcal{S}|$ received vectors using the matrix $\boldsymbol{G}$, each terminal recovers a total of$$(d-\ell-z+1)(|\mathcal{S}|-x)$$secure key symbols. The scheme delivers these secure symbols while requiring at most $d(d-\hat{d})+1$ transmissions over any edge of $\mathcal{G}$. Therefore, the resulting key rate is given by
$$R = \frac{(d-\ell-z+1)(|\mathcal{S}|-x)}{(d(d-\hat{d})+1)|\mathcal{S}|}.$$

\subsection*{Security Proof:}
The security proof follows directly from the security proof of Corollary \ref{C1}.

\end{proof}

We now remove the structural network condition used in Theorem~\ref{T3} and Corollary~\ref{C3}.

\begin{corollary}[Partially-connected single-source setting, without structural network requirements]\label{C2}
 Let $\mathcal{I}=\bigl(\mathcal{V},\mathcal{E},\mathcal{B},\mathcal{S}=\{s\}, \mathcal{T}=\{\mathcal{T}_i\}_{i=1}^m \bigr)$ be an instance of the Secure Multiple Key-Cast problem with
    \(\mathcal{B}=
    \bigl\{
    \beta \subseteq \mathcal{V} \setminus \{s\} : |\beta|\le \ell
    \bigr\},
    \)
    in which all terminal nodes are $d$-vertex connected from the source $s$. 
    Let $\hat{d}$ be the minimum vertex-connectivity from $s$ taken over all network nodes.
    Let \(\Omega \subseteq\mathcal{V}\) denote the set of partially-connected nodes.
    Define
    \[
    \mathcal{J}\triangleq \Big\{\, j\in\mathcal{V}\setminus \Omega\;:\; \big|\{u\in\Omega:(u\to j)\in\mathcal{E}\}\big| > z \Big\},
    \]
    i.e., $\mathcal{J}$ is the set of $d$-vertex connected nodes that receive input from more than $z$ partially-connected nodes.
    For \(j\in\mathcal{J}\), define
    \(
    p(j)\triangleq \big|\{u\in\Omega:(u\to j)\in\mathcal{E}\}\big|-z.
    \)
    Then, there exists a coding scheme that achieves a key-cast rate \(R\) with secrecy against any eavesdropper observing up to \(\ell\) nodes, where
    \[
    R \ge \frac{d-\ell-z+1}{d(d-\hat{d}) + \left(\sum_{j \in \mathcal{J}}\left\lceil \frac{p(j)}{d-\ell} \right\rceil\right) +1}.
    \]
    \end{corollary}
    
\begin{proof}
The proof follows the proof of Theorem~\ref{T3}, with modifications only for nodes in \(\mathcal{J}\).

Consider a \(d\)-vertex-connected node \(j\in\mathcal{J}\) attempting to recover its share \(\boldsymbol{s}_j^{(M)}\). 
Let \(\mathcal{J}_{PI} \subseteq \mathcal{P}(j)\) be the set of partially-connected parents of \(j\). Node \(j\) receives
\[
\boldsymbol{\tau}_{s\to j}^{(M)}
=
{\boldsymbol{v}_{j}^{(M)}}^T
\boldsymbol{M}
\boldsymbol{V}_{s\to j}^{(M)}
\]
from the source over its \(c_j\) incoming source edges. Note that if $j \notin \mathcal{P}(s)$, then \(c_j=0\). 

Next, let $\mathcal{P}_d(j) \triangleq \mathcal{P}(j)\setminus(\mathcal{J}_{PI}\cup\{s\})$ denote the set of $d$-vertex-connected parents of $j$. From each parent \(j_p\in\mathcal{P}_d(j)\), node \(j\) receives
\[
\tau_{j_p\to j}^{(M)}
=
{\boldsymbol{v}_{j}^{(M)}}^T
\boldsymbol{M}
\boldsymbol{v}_{j_p}^{(M)} .
\]

Since node $j$ is connected to more than $z$ partial-information nodes (which possess no information about $\boldsymbol{M}$), it requires $p(j)$ additional symbols of the form ${\boldsymbol{v}_j^{(M)}}^{T}\boldsymbol{M}\boldsymbol{v}'$ in order to recover its share $\boldsymbol{s}_j^{(M)} = \boldsymbol{M}\boldsymbol{v}_j^{(M)}$.

Because node \(j\) is \(d\)-vertex connected, Lemma \ref{d-vertex_c} implies that the source can securely transmit the row vector
\begin{equation}\label{C1eq2}
    {\boldsymbol{v}_j^{(M)}}^{T}\boldsymbol{M}\boldsymbol{\hat{V}}_{j}^{(M)}
\end{equation}
to node \(j\). Here, \(\boldsymbol{\hat{V}}_{j}^{(M)}\) is a \((d-z)\times p(j)\) Vandermonde matrix whose columns correspond to distinct Vandermonde indices and are chosen to be linearly independent of the columns of both \(\boldsymbol{V}_{s\to j}^{(M)}\) and \(\{\boldsymbol{v}_{j_p}^{(M)}\}_{j_p \in \mathcal{P}_d(j)}\). 
Delivering these $p(j)$ symbols requires \(\left\lceil \frac{p(j)}{d-\ell} \right\rceil\) additional transmissions because, by Lemma \ref{d-vertex_c}, each transmission can securely send at most \(d-\ell\) symbols in the presence of an adversary observing up to \(\ell\) nodes.

Node \(j\) then aggregates the symbols received from the source, its $d$-vertex-connected parents, and the additional secure transmissions from \eqref{C1eq2} to construct the combined row vector:
\begin{equation}\label{C1eq1}
   \boldsymbol{T}_{j}^{(M)}
   =
   {\boldsymbol{v}_{j}^{(M)}}^{T}
   \boldsymbol{M}
   \boldsymbol{V}_{j}^{(M)} ,
\end{equation}
where the combined observation matrix $\boldsymbol{V}_{j}^{(M)}$ is formed as
\[
\boldsymbol{V}_{j}^{(M)}
\triangleq
\left[\,
\boldsymbol{V}_{s\to j}^{(M)}
\;,\;
\boldsymbol{\hat{V}}_{j}^{(M)}
\;,\;
\boldsymbol{V}_{\mathcal{P}_d(j)}^{(M)}
\,\right],
\]
and $\boldsymbol{V}_{\mathcal{P}_d(j)}^{(M)}$ is the matrix whose columns are $\{\boldsymbol{v}_{j_p}^{(M)}\}_{j_p \in \mathcal{P}_d(j)}$. 

Because \(\boldsymbol{V}_{j}^{(M)}\) is constructed as a \((d-z) \times (d-z)\) Vandermonde matrix with distinct columns, it is full rank and invertible. Therefore, node $j$ can reconstruct its target share as:
\[
\boldsymbol{s}_j^{(M)}
=
\Big(
\boldsymbol{T}_{j}^{(M)}
{\boldsymbol{V}_{j}^{(M)}}^{-1}
\Big)^{T}
=
\boldsymbol{M}\boldsymbol{v}_{j}^{(M)}.
\]

Applying this procedure to all nodes in \(\mathcal{J}\), each node can reconstruct its share using $\left\lceil \frac{p(j)}{d-\ell} \right\rceil$ additional secure transmissions from the source. Therefore, the achieved rate satisfies
\[
R \ge
\frac{d-\ell-z+1}{
d(d-\hat{d}+1)
+ \left(\sum_{j \in \mathcal{J}} \left\lceil \frac{p(j)}{d-\ell} \right\rceil \right)+ 1},
\]
which is strictly positive whenever \(\mathcal{J}\) is finite.
\end{proof}

\section{Acknowledgements}
This work is supported in part by NSF grant CCF-2245204.

\bibliographystyle{unsrt}
\bibliography{references}

\end{document}